\documentclass{article}
\usepackage{amsmath}
\usepackage{amssymb}
\usepackage{latexsym} 
\usepackage[all]{xypic}
\usepackage{listings}
\usepackage{amsmath,amsfonts,amssymb}
\usepackage{mathrsfs}
\usepackage{amsmath}
\usepackage{amssymb}
\usepackage{enumerate}
\usepackage{eufrak}
\usepackage{fancyhdr}
\usepackage[dvips,breaklinks=true]{hyperref}
\usepackage[all]{xypic}
 \usepackage[utf8]{inputenc}
\usepackage{mytheorems}
\input macros.tex
\date{}
\title{Undirected Graphs of Entanglement 3}
\author{Walid Belkhir \\
\begin{small}walid.belkhir@lif.univ-mrs.fr\end{small}\\
  \begin{small} Laboratoire d'Informatique Fondamentale de Marseille.\end{small} \\
\begin{small}  Université de Provence.  \end{small}
}

\begin{document}

\maketitle

 \begin{abstract} Entanglement is a complexity measure of digraphs that origins in fixed-point logics. Its combinatorial  purpose is to measure the nested depth of  cycles in digraphs. We address the problem of  characterizing the structure of graphs of  entanglement at most $k$. Only partial results are known so far:  digraphs for $k=1$, and undirected  graphs for $k=2$.     In this paper we investigate the structure of undirected graphs for $k=3$. Our main tool is the so-called \emph{Tutte's decomposition} of $2$-connected graphs into cycles and $3$-connected components into a tree-like fashion. We shall give necessary conditions on Tutte's tree to be a tree  decomposition of a  $2$-connected graph  of entanglement $3$.     
\end{abstract}

\begin{small}\textbf{Key words.} Entanglement, tree decomposition, Tutte's decomposition, connectivity.
\end{small}

\section{Introduction}
Entanglement is a complexity measure of digraphs introduced in \cite{berwanger} to analyze the descriptive complexity of the propositional modal $\mu$-calculus \cite{kozen}.  This measure has shown its use in  solving the variable hierarchy problem\footnote{This problems asks whether the expressive power of a given fixed-point logic increases with the   number of bound variables.} for the modal $\mu$-calculus \cite{BerwangerGraLen06} and for the lattice $\mu$-calculus \cite{LPAR08}. However, little is know about the graph theoretic and algorithmic properties of entanglement, such that the structural properties of graphs of bounded entanglement. The structure of digraphs of entanglement $\le 1$ is known \cite{berwanger}. Using this result, it was argued that deciding whether a digraph has entanglement $\le 1$ is a problem in NLOGSPACE. The  graphs\footnote{Throughout this paper, an undirected graph is called simply a graph; a directed graph is called a digraph.} of entanglement $\le 2$ have been characterized in \cite{BelkSanto0a7}. Using this characterization,   a linear time algorithm recognizing those  graphs has been devised.       

 In this paper we study the structure of  graphs of entanglement $3$. Our approach consists in studying entanglement w.r.t. the notions of connectivity and cyclicity. The latter  are among the basic  properties of graphs: a graph may be decomposed into maximal articulation point-free components in a tree like fashion.  An articulation point-free component itself admits  a  canonical decomposition  in terms of cycles and $3$-connected components  in a tree like structure, this decomposition  was given by Tutte \cite{TutteConnect}. Continuing this process,   Hohberg   gave   a  decomposition  of graphs of higher connectivity \cite{walter91}, this is  a sort of generalization of Tutte's decomposition.  However, Hohberg's decomposition  is  no longer canonical.  Another decomposition of digraphs in terms of directed cycles, known as \emph{ear} decomposition \cite[\S 7]{digraphs}, has been useful  for studying the connectivity of directed mutligraphs.     Such decompositions   are important tools in graph theory and they  are often used  in inductive proofs and constructions.  For instance see \cite{RichterEmbedding} where  Tutte's decomposition is used. 

In \cite{BelkSanto0a7} it was pointed out that the canonical decomposition of graphs into  $2$-connected components  is well  adapted in  characterizing  the structure of graphs of entanglement $\le 2$.   With a similar approach we are asking whether Tutte's decomposition would be of use to characterize  the structure of  graphs of entanglement $3$.  Roughly speaking Tutte's decomposition reduces the structure of a given graph into that of its components. So what can we  say about the components  themselves if the given graph has entanglement $3$ ?  We  shall find necessary conditions on Tutte's tree to be a tree decomposition of  a $2$-connected  graph  of entanglement $3$: we shall \emph{(i)}  characterize the structure of $3$-connected components of $G$,  \emph{(ii)}  impose necessary conditions on the manner by  which those components are glued together, and \emph{(iii)}  find an upper bound of the diameter of Tutte's tree.

  We suggest a  natural path  for generalizing the algebraic approach used  to  construct  the  graphs of entanglement $2$ \cite{BelkSanto0a7}  to the construction of  the graphs of entanglement $3$. We mean by the algebraic approach the manner by which a class of graphs is build up out of  small pieces of graphs using an   appropriate set of  gluing operations.   This  generalization   interacts with Tutte's decomposition  in a  deep  and not obvious way. Finally, ear decomposition would be useful in dealing with the entanglement within the direct setting. Furthermore, Hohberg's decomposition suggests a hopeful path for the characterization  of   graphs of higher  entanglement.

 \cutout{We shall examine the Tutte's decomposition and we shall adapt it to provide a decomposition of undirected graphs of entanglement $\le 3$. }

\section{Preliminaries}
A directed graph $G=(V_G,E_G)$ is a set  of vertices $V_G$ and a binary relation $E_G \subseteq V_G \times V_G$.
All the graphs in this paper are  finite  and undirected until we say otherwise.  They are also  simple i.e. without multi-edges, but during the decomposition process\footnote{Indeed we mean Tutte's decomposition.} some mutli-edges may  appear.  If $X\subseteq V_G $, then we shall write $G[X]$ for the subgraph of $G$ induced by $X$. 

\subsection{Cyclicity}
A \emph{feedback vertex set}  \cite{Festa:FeedBack} of a digraph $G$ is a subset $X \subseteq V_G$ that meets all directed cycles of $G$, i.e. the digraph  $G\setminus X$ is acyclic. The \emph{cyclicity} of $G$, denoted $\Cycl{G}$, is the cardinality of the \emph{minimum} feedback vertex set. \\
 When we deal with the   cyclicity of a graph $G$  we consider  the \emph{directed} cycles of the \emph{symmetric directed} graph $G$. In other words, every (undirected) edge $v_1v_2$ of $G$ may be viewed as the set of (directed) edges $\set{(v_1,v_2),(v_2,v_1)}$, and therefore $v_1v_2v_1$ is  a directed cycle of length $2$.
 To make the notion of  the cyclicity of a graph easier and more  intuitive, it is convenient to see the feedback vertex set of a  graph as an \emph{edge cover set}, i.e. a set $X \subseteq V_G$ is an edge cover of $G$ if for each edge $v_1v_2 \in E_G$ we have that $v_1\in X$ or $v_2 \in X$.  

\begin{lemma}
Let $G$ be an undirected graph, and let $X \subseteq V_G$. Then, $X$ is a feedback vertex set of $G$  if and only if $X$ is an edge cover set of $G$.    
\end{lemma}
\begin{proof} $\textrm{ } $ \\
$\implies$\\
 Let $X\subseteq V_G$ be a feedback vertex set of $G$. Therefore, for every cycle of  the form $v_1v_2v_1$ we have that $v_1 \in X$ or $v_2 \in X$. Since the cycle $v_1v_2v_1$ is just the edge $v_1v_2 $ of the undirected graph  $G$,  then  $X$ is an edge cover set of $G$. 
\\
$\Longleftarrow$ \\  Let $X \subseteq V_G$ be an edge cover set of $G$ and let  $C_n=v_1v_2\dots v_nv_1$ be a cycle of $G$. Then,  for each  edge $v_iv_{i+1}$ of $C_n$  we have $v_i\in X$ or $v_{i+1} \in X$, because $X$ is an edge cover. Therefore $X$ meets  $C_n$, hence  $X$ is a feedback vertex set of $G$.    
\end{proof}
It follows by the previous Lemma that  the cyclicity of an undirected graph  $G$  is the cardinality of the minimum edge cover set  of $G$. 

\subsection{Connectivity} 
Given two sets $A,B$ of vertices, we call $\pi=v_1\dots v_n$ an $A$-$B$ \emph{path} if $V_{\pi} \cap A=\set{v_1}$ and $V_{\pi} \cap B=\set{v_n}$. Two or more paths are \emph{independent} if none of them contains an internal vertex of the other. If $a,b$ are two vertices then  we shall write ''$a$-$b$ path'' instead of "$\set{a}$-$\set{b}$ path''.  In this case,  two $a$-$b$ paths are independent if and only if $a$ and $b$ are their only common vertices.
\begin{definition}
A  graph is $k$-connected \cite[\S 3]{GraphTheoryBookRein} if one needs to remove at least $k$-vertices to disconnect it. The \emph{connectivity} of a graph $G$ is the maximum $k$ such that $G$ is $k$-connected.
\end{definition}
 The following  Theorem, due to Menger \cite[\S 3.3]{GraphTheoryBookRein}, provides a useful equivalent definition of $k$-connectivity.
\begin{theorem}\label{Menger:Theorem}
 A graph is $k$-connected if and only if it contains $k$ independent paths between any two vertices.
\end{theorem} 
A vertex whose removal disconnects the graph is called \emph{articulation point}. A maximal connected subgraph without articulation points is called  a \emph{biconnected component}.  By their maximality, the biconnected components of $G$ overlap in at most one vertex, which is an articulation point of $G$. Hence, every edge of $G$ belongs to a unique biconnected component. An isolated vertex   is considered trivially   as a biconnected component.   Let $\mathcal{A}$ be the set of articulation points of $G$, and $\mathcal{B}$ the set of its $2$-connected components. Then we  obtain a bipartite graph  $T_1(G) $ whose vertices are $\mathcal{A} \cup \mathcal{B}$ and whose edges are  $aB$ whenever $a \in B$.  It is one of the elementary results to  show that if $G$ is connected then $T_1(G)$ is a tree, see Proposition 3.1.2 of \cite{ GraphTheoryBookRein}.
  \cutout{ We say that $G$ is \emph{$k$-cyclic} if $G$ is strongly connected and $\Cycl{G}=k$, for $k\ge 1$. }
 \cutout{ If $W \subseteq V_G$, we denote by $G[W]$ the sugraph of $G$ induced by $W.$}
\subsection{Separations, hinges}
In this subsection we introduce the definition of separations and hinges. The material presented in this here  can be found in  \cite{TutteConnect}, whereas  we shall adopt the notation and the terminology  of  \cite{Droms2Connect} and \cite{Richter04}.\\
A \emph{bond} is a graph consisting of just two vertices and at least one edge.  A \emph{$k$-bond} is a bond of $k$ edges. Observe that, if $k\ge 2$, then a $k$-bond is a multi-graph.
If $A,B \subseteq V_G$ and $S\subset V_G$ are such that every  $A$-$B$ path in $G$ contains a vertex of $S$, then we say that $S$ separates the sets $A$ and $B$ in $G$. Observe that this implies that $A\cap B \subseteq S$. We shall say that $S$ separates $G$ if $G \setminus S$ is disconnected, that is, if $S$ separates $G$ into some vertices which are not in $S$. A separating set of vertices is called a \emph{separator}.  
A   pair $(A,B)$ is  a \emph{separation} of  $G$  if  $A \cup B=V_G$ and $G$ has no edge between $A\setminus B$ and $B\setminus A$.  Clearly, this is equivalent to saying that $A\cap B$ separates $A$ from $B$. We say that a  separation  $(H,K)$ is a \emph{$k$-separation} if $\card{V_{H\cap K}}=k$. A subgraph $K$ of $G$ is an $H$-\emph{bridge} if $K$ is obtained from a component $\mathcal{C}$ of $G \setminus V_{H}$ by adding to $\mathcal{C}$ all the edges of $G$ which have at least one end in $\mathcal{C}$.  From the above definition of connectivity follows  a standard result of graph theory \cite{TutteConnect}:
\begin{lemma}
 A graph is $k$-connected if there is no $m$-separation, for all $m=0,\dots,k-1$.
\end{lemma}

  From the  definition of connectivity, it follows  that the connectivity of the following graphs is infinite:  the graph of just one vertex, the $k$-bonds, for $k\ge 1$,  and the $3$-clique, see Figure \ref{infiniteconnectivity}.
\begin{figure*}[h]
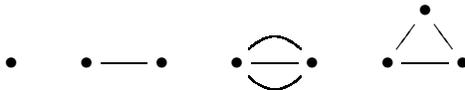

$$
  \xygraph{ %
   !{<0cm,0cm>;<0.5cm,0cm>:<0cm,0.7cm>::}
    []*+{\bullet}="v0"
    [rr]*+{\bullet}="v1"
    [rr]*+{\bullet}="v2"
    [rr]*+{\bullet}="v3"
    [rr]*+{\bullet}="v4"
( "v3"[r(1)][u(1)]="v3v4U",
"v3"[r(1)][d(1)]="v3v4D",)
    [rr]*+{\bullet}="v5"
    [rr]*+{\bullet}="v6"
    [lu]*+{\bullet}="v7"
"v1"-"v2"
"v3"-"v4"
"v5"-"v6"
"v6"-"v7"
"v5"-"v7"
"v3"-@{-}@`{"v3v4U"} "v4"
"v3"-@{-}@`{"v3v4D"} "v4"
}
$$
\caption{Graphs with infinite connectivity.}
\label{infiniteconnectivity}
\end{figure*}  
 \begin{convention}\label{Conn:Clique:Conv}
 For technical reasons, we shall consider  the connectivity of the $k$-cliques, for $k\ge 3$, equals  $k-1$.  
 \end{convention}

If $(A,B)$ is a $2$-separation of $G$, then we say that $A \cap B$  is a \emph{hinge} of $G$ iff  at least one of $G[A]$ and $G[B]$ is $2$-connected.

\subsection{Graph minors}
 A graph $G$ is \emph{minor}  of a graph $H$  if $G$ can be obtained from $H$ by successive application of the following operations on it:    
\begin{enumerate}[({a})]
\item  delete an edge,
\item  contract an edge, and
\item delete an isolated vertex.
\end{enumerate}
 A class $\classe$ of graphs is  minor closed  if is closed under taking minors, that is whenever $H \in \classe$  then for each minor $G$ of $H$ we have that $G\in \classe$.

\subsection{Tree decomposition}
We reproduce here the definition of the notion of \emph{tree decomposition} as given by Robertson and Seymour \cite{GraphMinors2}.   
\begin{definition}
Let $G$ be a graph, $T$ be a tree, and let $X=(V_t)_{t\in T}$ be a family of subsets of $V_G$. We say that the pair $(T,X)$ is a \emph{tree decomposition} of $G$ if the following conditions hold:
\begin{enumerate}[({T}-1)]
\item $V_G = \bigcup_{t\in T} V_t$,
\item for every edge $v_1v_2 \in E_G$ there exists $V_t$ such that $v_1,v_2 \in V_t$,
\item if  there is a path  $t_1\dots t_2 \dots t_3$ in $T$ then $V_{t_1} \cap V_{t_3} \subseteq V_{t_2}$.
\end{enumerate}
\end{definition}
Conditions (T-1) and (T-2) say the the graph $G$ is the union of subgraphs  induced by the set of vertices $V_t$; the sets $V_t, t \in V_T$ as well as the subgraphs induced by $V_t, t\in V_T$ are called the \emph{bags}  of  the tree decomposition. Condition (T-3) states that the bags of the tree decomposition are organized into a tree like fashion.\\

One of the most important feature of the  tree decomposition concept is that it shows a natural correspondence between the properties of the separations of the graph and its tree decomposition:
\begin{lemma}\label{TreeDecomp:Separation:Lemma}
Let $G$ be a graph and let $(T,{(V_t)}_{t\in T})$ be a tree decomposition of $G$. Given an edge $t_1t_2$ of $ T$ and  let $T_1,T_2$ be the components of $T \setminus t_1t_2$ such that $t_1 \in T_1$ and $t_2 \in T_2$. \\
Then $V_{t_1}\cap V_{t_2}$ separates $U_1:=\cup_{t\in T_1} V_t$ from   $U_2:=\cup_{t\in T_2}V_t$ in $G$.
\end{lemma}   
\begin{proof}
Since $T$ is a tree, then every $t$-$t'$ path in $T$ with $t\in T_1$ and $t' \in T_2$ contains both $t_1$ and $t_2$. Therefore, it follows by (T-3) that  $U_1 \cap U_2 \subseteq V_{t_1} \cap V_{t_2}$. To accomplish the proof it remains to show that  $G$ does not contain an edge $u_1u_2$ with $u_1 \in U_1 \setminus U_2$ and $u_2\in U_2 \setminus U_1$. If such an edge exists then it follows from (T-2) that there exists $t \in V_T$ such that $u_1,u_2 \in V_t$. By assumption, we have chosen $u_1$ in $U_1 \setminus U_2$ and hence $t \in T_1$. Also, we have chosen $u_2$ in $U_2\setminus U_1$ and hence $t \in T_2$. This is a contradiction because by the hypothesis we have $T_1 \cap T_2=\emptyset$.  
\end{proof}
A \emph{bag}, denoted $\beta_t$, $t \in V_T$, is the subgraph of $G$ induced by vertices $V_t$.
A  \emph{torso}\footnote{a torso and not  a torsor. The plural  of torso is torsos. We adopt the terminology of \cite{GraphTheoryBookRein}.}, denoted $\tau_t$, is the bag $\beta_t$  where we add an edge $vv'$ to the multiset of edges of $\beta_t$ for each  $v,v' \in V_t \cap V_{t'}$ such that  $tt' \in E_T$.   Observe that after adding such edges the torso may become a multigraph even if the starting graph is simple.  However, we can split the edges of a torso into two sets, the set of the original edges which belongs to the bag, and the set of edges which we have added,   the latter are called the \emph{virtual} edges.

\section{Tutte's Theorem}
The following Theorem, known as \emph{Tutte decomposition Theorem} -- which provides a   tree decomposition of $2$-connected graphs into cycles, bonds and $3$-connected components -- will be our main working tool. 
  
\begin{theorem}\label{Tutte:Theorem}\cite{TutteConnect}
Every $2$-connected graph has a tree decomposition $({T,(V_t)}_{t\in T})$ such that $\card{V_t \cap V_{t'}} =2$ for every $tt' \in E_T$, and moreover every torsos is either a $k$-bond ($k \ge 3$), or  $3$-connected or a cycle. Conversely, every graph with such a tree decomposition is $2$-connected. Furthermore such a decomposition is unique.
\end{theorem} 

The key idea behind Tutte's Theorem consists in considering a particular  set of $2$-separations of $G$, these are the $2$-separations which are \emph{compatible}  with all the other  $2$-separations. Two separations $(U_1,U_2)$ and $(W_1,W_2)$ are compatible if we can find $i,j \in \set{1,2}$ such that $U_i \subseteq W_j$ and $U_{3-j} \supseteq W_{3-j}$.   It follows from the proofs of Tutte's Theorem presented in the literature \footnote{For instance that of  \cite{Richter04},  which uses the notion of hinges,  and of \cite{GraphTheoryBookRein}, that uses the notion of compatible separations.}  that $(U_1,U_2)$ is a separation which is compatible with all the other separations if and only if $U_1 \cap U_2$ is a hinge. Finally Tutte's decomposition arises from considering a tree decomposition $(T,(V_t)_{t\in T})$ such that $V_t \cap V_{t'}$ is a hinge for all $tt' \in E_T$.   

The original proof of this  Theorem can be found in \cite{TutteConnect}. An  extension of this Theorem to locally finite\footnote{An infinite graph is locally finite if the degree of each  vertex is finite.} graphs is referenced in \cite{Droms2Connect}.  A generalization of this Theorem to arbitrary infinite graphs  can be found in \cite{Richter04}.
 In the sequel we shall refer to such a decomposition as \emph{Tutte's decomposition}.

An example of a graph and its decomposition tree is depicted in Figure \ref{TutteTreeExample1}.1. Let $U_1:=V_G \setminus \set{v_9}$ and $U_2:=\set{v_6,v_8,v_9}$.  Observe that $(U_1,U_2)$ is a $2$-separation of $G$.  The set $U_1 \cap U_2=\set{v_6,v_8}$ is a $2$-separator of $G$ but it is not a hinge because neither the subgraph   $G[U_1]$ nor $G[U_2]$ is $2$-connected.   By inspecting the set of $2$-separators, we can deduce that the  set of hinges is   $\set{\set{v_1,v_2},\set{v_3,v_4},\set{v_5,v_6},\set{v_8,v_9}}$.
The virtual edges arising from this decomposition   are represented by  dashed lines.    

\begin{figure}
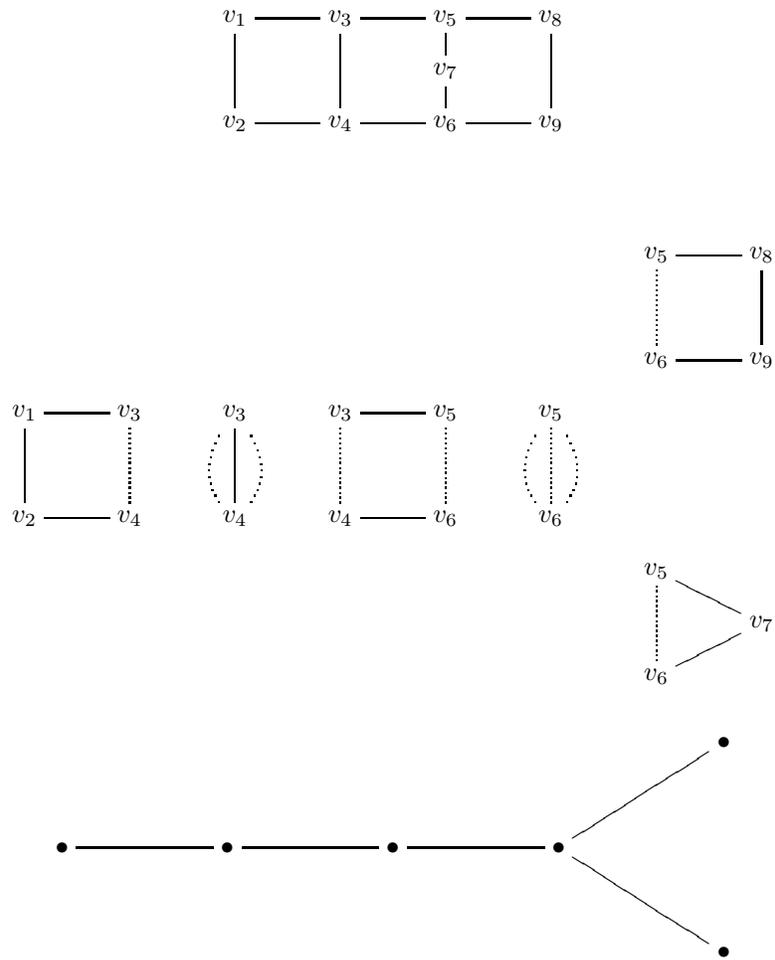

$$
  \xygraph{ 
   !{<0cm,0cm>;<0.7cm,0cm>:<0cm,0.7cm>}
    []*+{v_1}="v1"
    [rr]*+{v_3}="v3"
    [rr]*+{v_5}="v5"
    [d]*+{v_7}="v7" 
    [urr]*+{v_8}="v8"
     [dd]*+{v_9}="v9"
    [ll]*+{v_6}="v6"
    [ll]*+{v_4}="v4"
    [ll]*+{v_2}="v2"
    "v1"-"v3"  
    "v3"-"v5"
  "v5"-"v8"
  "v8"-"v9"
  "v9"-"v6"
  "v6"-"v7"
    "v7"-"v5"  
    "v6"-"v4"
  "v4"-"v3"
  "v4"-"v2"
  "v2"-"v1"
}
$$ 

$$
\textrm{ }
$$

$$
  \xygraph{ %
   !{<0cm,0cm>;<0.7cm,0cm>:<0cm,0.7cm>::}
    []*+{v_1}="v1"
    [rr]*+{v_3}="v3"
    [dd]*+{v_4}="v4"
    [ll]*+{v_2}="v2" 
    [rrrr]*+{v_4}="w4"
     [uu]*+{v_3}="w3"
    [rr]*+{v_3}="z3"
    [dd]*+{v_4}="z4"
   ("w3"[l(1)][d(1.2)]="w3w4G",
      "w3"[r(1)][d(1.2)]="w3w4D",)
 [rr]*+{v_6}="v6"
   [uu]*+{v_5}="v5"
   [rr]*+{v_5}="w5"
   [dd]*+{v_6}="w6"
   ("w5"[l(1)][d(1.2)]="w5w6G",
      "w5"[r(1)][d(1.2)]="w5w6D",)
   [rruuu]*+{v_6}="z6"
   [uu]*+{v_5}="z5"    
   [rr]*+{v_8}="v8"
  [dd]*+{v_9}="v9"
  [ddddd]*+{v_7}="v7"
  [llu]*+{v_5}="zz5"
  [dd]*+{v_6}="zz6"
"v1"-"v3"
"v3"-@{..}"v4"
"v4"-"v2"
"v2"-"v1"
"w3"-"w4"
"w3"-@{..}@`{"w3w4G"}"w4"
"w3"-@{..}@`{"w3w4D"}"w4"
"z3"-@{..}"z4"
"z3"-"v5"
"v5"-@{..}"v6"
"v6"-"z4"
"w5"-@{..}"w6"
"w5"-@{..}@`{"w5w6G"}"w6"
"w5"-@{..}@`{"w5w6D"}"w6"
   "z5"-"v8"
   "v8"-"v9"
   "v9"-"z6"
   "z5"-@{..}"z6"
   "zz5"-"v7"
   "zz6"-"v7"
   "zz5"-@{..}"zz6"
} $$
$$
\textrm{  }
$$
  $$
\xygraph{ %
   !{<0cm,0cm>;<1.1cm,0cm>:<0cm,0.7cm>::}
    []*+{\bullet}="t1"
    [rr]*+{\bullet}="t2"
    [rr]*+{\bullet}="t3"
     [rr]*+{\bullet}="t4"
    [rruu]*+{\bullet}="t5"
    [dddd]*+{\bullet}="t6"
"t1"-"t2"
"t2"-"t3"
"t3"-"t4"
"t4"-"t5"
"t4"-"t6"
}
$$
\label{TutteTreeExample1}
 \caption{A graph with its Tutte's decomposition tree}
 \end{figure}

The following Lemma provides some useful Properties of Tutte's Tree.  
\begin{lemma}\label{Prop:Tree}
Let $G$ be a $2$-connected graph and $T$ its Tutte's tree. Then,
\begin{enumerate}[(a)]
\item  $T$ is a bipartite graph $(T_1,T_2)$  where $t\in T_1$ if and only if $V_t$ is a hinge. 

\item  Let $t_1 \dots t_n$ be a path in $T$ then,   if $h_i$ denotes the hinge $V_{t_i}$ then
\begin{enumerate}[({b.}1)]
\item for all $i=1,\dots,n -1$, we have either $h_i \cap h_{i+1}=\emptyset$ or $|h_i \cap h_{i+1}|=1$, 
\item    if $h_p \cap h_q =\emptyset$ where $1\le p <q\le n$,  then for all $i \le p$ we have that $h_i \cap h_q=\emptyset$, and 
\item if $h_p \cap h_q =\emptyset$ where $1\le p <q\le n$,  then for all $i \ge q$ we have that $h_p \cap h_i=\emptyset$.
\end{enumerate}
\end{enumerate}
\end{lemma}
\begin{proof}
The statement (a) follows from  Tutte's Theorem \ref{Tutte:Theorem}. 
The statements (b.1), (b.2), and (b.3) are a direct consequence of  Lemma \ref{TreeDecomp:Separation:Lemma}. 
\end{proof}
\begin{lemma}\label{howmanybridges}\cite{Richter04}. 
Let $G$ be $2$-connected graph. A  $2$-separator $\set{x,y}$ is a hinge in $G$ if and only if \emph{(i)} either there are at least three  $[x,y]$-bridges, or \emph{(ii)} there are two  $[x,y]$-bridges at least one of them is $2$-connected. 
\end{lemma}
 
  Tutte's decomposition tree  of a given $2$-connected graph gives rise to an algebraic expression in terms of $k$-bonds, cycles and $3$-connected components and the $2$-Sum operator. The latter operation consists in  taking two graphs, choosing a $2$-clique from each, identifying the vertices in the cliques and deleting the edges of the cliques, see Figure \ref{twosumoperator}.

The $k$-sum operators on graphs,  for $k=1,2,3$,\footnote{The $k$-sum, for $k=3$ is defined in a similar way of the $2$-Sum operator apart that we take a $3$-click of each graph. } have been introduced in \cite{Wagner37} in the aim to prove a theorem which states that a graph which  does not contain a $K_5$ as a minor may be expressed by means of these operations starting with the class of planar graphs and a particular graph  on  $8$ vertices\footnote{This Theorem has been extended to matroids in \cite{SomeAppSeymour} by extending these  operations to   matroids as well.}.
     
\begin{figure*}[h]
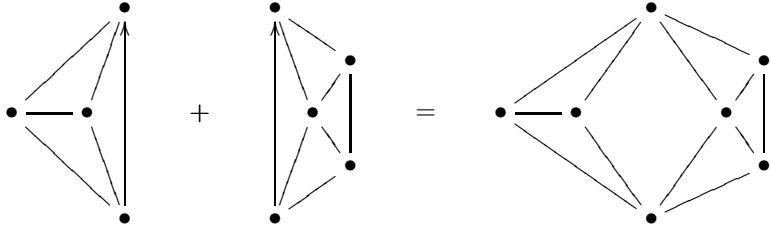

  $$
  \xygraph{ %
   !{<0cm,0cm>;<0.5cm,0cm>:<0cm,0.7cm>::}
    []*+{\bullet}="v0"
    [dddd]*+{\bullet}="v1"
    [uul]*+{\bullet}="v2"
   [ll]*+{\bullet}="v3"
"v2"[rrr]+{+}="plus"
 [rrr]*+{\bullet}="w2"
 [luu]*+{\bullet}="w0"
[dddd]*+{\bullet}="w1"
[rru]*+{\bullet}="w4"
[uu]*+{\bullet}="w3"
[drr]*+{=}="egal"
[rr]*+{\bullet}="z3"
[rr]*+{\bullet}="z2"
[rruu]*+{\bullet}="z0"
[dddd]*+{\bullet}="z1"
[uurr]*+{\bullet}="zz2"
[ru]*+{\bullet}="zz3"
[dd]*+{\bullet}="zz4"
"v1":"v0"
"v3"-"v0"
"v3"-"v2"
"v3"-"v1"
"v2"-"v0"
"v2"-"v1"
"w1":"w0"
"w2"-"w0"
"w2"-"w1"
"w2"-"w3"
"w2"-"w4"
"w0"-"w3"
"w1"-"w4"
"w3"-"w4"
"z3"-"z0"
"z3"-"z2"
"z3"-"z1"
"z2"-"z0"
"z2"-"z1"
"zz2"-"zz3"
"zz2"-"zz4"
"z1"-"zz4"
"zz3"-"zz4"
"z0"-"zz3"
"z0"-"zz2"
"z1"-"zz2"
}
$$
\caption{The $2$-Sum operator}
\label{twosumoperator}
\end{figure*}

\begin{definition} \label{twosum:def}
Let $G_1$ and $G_2$ be graphs and let $e_1=(a_1,b_1)  \in E_{G_1} $ and  $e_2=(a_2,b_2)\in  E_{G_2}$. Define the $2$-Sum of $G_1$ and $G_2$  on $e_1$ and $e_2$ respectively, denoted $G_1 +_{e_1e_2} G_2$,  to be the graph obtained from
 the  union  of $G_1$ and $G_2$ by identifying the vertex $a_1$ with $a_2$ and the vertex $b_1$ with $b_2$ and deleting the edges $e_1$ and $e_2$, see Figure \ref{twosumoperator}. 
\end{definition}

It is clear that (see \cite{TutteConnect}):
\begin{lemma}\label{sum:lemma:bicconcted}
 $G$ and $H$ are $2$-connected if and only if $G +{_{e_1e_2}} H$ is $2$-connected, for each $e_1 \in E_G, e_2 \in E_H$.
\end{lemma}

Note that the $2$-bond is the neutral element of the $2$-Sum. From the previous Lemma it follows that:

\cutout{If $G$ is $k$-separable, $i=1,2$, then $G$ is expressible as a $k$-sum. When $k=1$, then $G$ is clearly expressible as $1$-sum of its $2$-connected components. }

\begin{proposition}
Let $\mathcal{B}_2$ be the least class of graphs containing the $k$-bonds ($k\ge 3$), the cycles, the $3$-connected graphs and closed under the $2$-Sum operator. Then, $G\in \mathcal{B}_2$  if and only if $G$ is $2$-connected.  
\end{proposition}

\begin{lemma}\label{torsos:minor:Coro}
Let $G$ be $2$-connected and let  $({T,(V_t)}_{t\in T})$ be its Tutte decomposition, then  every torso  $\tau$ of $T$ is a minor of $G$.
\end{lemma}
\begin{proof}
Let $t \in T$ and let $t_1,\dots,t_n$ be the neighbours of $t$ in $T$.  Let $T_i,i=1,\dots,n$ be the component   of $T\setminus tt_i$ that contains $t_i$, and let $\overline{T}_i$ be the complement of $T_i$ w.r.t $T$. Let also $U_i:= \cup_{t\in T_i}V_t$. Recall that  if $U \subseteq V_G$ then the subgraph of $G$ induced by $U$ is denoted by $G[U]$.  \\
 We shall prove that the torso $\beta_t$ is a minor of $G$.  On the one hand, since the subgraph  $G[U_i]$ is connected, see Theorem IV.20  of \cite{TutteGraphTh}, then there exists an $u_i$-$w_i$ path $\pi_i$ in $G[U_i]$ such that $\set{u_i,w_i}=V_t\cap V_{t_i}$. On the other hand, since $\set{u_i,w_i}$ is a $2$-separator in  $G$ of $U_i$ from $\overline{U}_{i}$ by Lemma \ref{TreeDecomp:Separation:Lemma}, then $\set{u_i,w_i}$ is also a $2$-separator in $G$ of $U_i$ from $U_j$ where $i \neq j$ because $U_j \subseteq \overline{U}_i$. Therefore, the $u_i$-$w_i$ paths, for $i=1,\dots,n$, are independent. Observe that the graph $\beta_t^{+}$  consisting of the bag $\beta_t$ where we  add the $u_i$-$w_i$ paths is a subgraph of $G$. Since the $u_i$-$w_i$ paths are independent we can contract each path until we get a single edge and moreover the resulting graph is just the torsos $\tau_t$. Hence $\tau_t$ is a minor of $\beta_t^{+}$, and $\beta_t^{+}$ is a subgraph of $G$, hence $\tau_t$ is a minor of $G$. 
\cutout{Recall the edges of $\beta_t$ are the edges of the bag $\beta_t$  we have add an edge $e_i$ for each neighbor $t_i$.}
\end{proof}

From the  previous Lemma we get the following particular case:
\begin{lemma}\label{Minor:2Sum}
Let $G$ be a $2$-connected graph, such that $G=G_1 +_{e_1e_2}G_2$. Then, $G_i+e_i$ is a minor of $G$, $i=1,2$. 
\end{lemma}

\section{Entanglement, connectivity, and  edge covering  }
The entanglement of a finite digraph $G$, denoted $\Ent{G}$, was
defined in \cite{berwanger} by means of some games $\Ent{G,k}$, $k =
0,\ldots ,\card{V_{G}}$. The game $\Ent{G,k}$ is played on the digraph
$G$ by Thief against Cops, a team of $k$ cops. The rules are as
follows. Initially all the cops are placed outside the digraph, Thief
selects and occupies an initial vertex of $G$.  After Thief's move,
Cops may do nothing, may place a cop from outside the digraph onto the
vertex currently occupied by Thief, may move a cop already on the
digraph to the current vertex.  In turn Thief must choose an edge
outgoing from the current vertex whose target is not already occupied
by some cop and move there.  If no such edge exists, then Thief is
caught and Cops win.  Thief wins if he is never caught.  The
entanglement of $G$ is the least $k \in N$ such that $k$ cops have a
strategy to catch the thief on $G$. It will be useful to formalize
these notions.

\begin{definition}\label{entang:def1}
  The entanglement game $\Ent{G,k}$ of a digraph $G$ is defined by:
\begin{itemize}
  \item Its positions are of the form $(v,C,P)$, where $v \in V_{G}$,
    $C \subseteq V_{G}$ and $\card{C} \leq k$, $\monespace{2mm} P \in \{Cops,
    Thief\}$.
 \item  Initially Thief chooses $v_{0} \in V_G$ and moves to
    $(v_0,\emptyset,Cops)$. 
 \item Cops can move  from $(v,C,Cops)$ to $(v,C',Thief)$
    where $C'$ can be
     \begin{enumerate}[$\bullet$]
     \item   $C$ : Cops skip,
       \item $C \cup\set{v}$ : Cops add a new Cop on the
      current position,
     \item  $(C \setminus\set{x}) \cup \set{v}$ : Cops move a placed Cop
      to the current position.
\end{enumerate}
 \item Thief can move from $(v,C,Thief)$ to $(v',C,Cops)$ if
    $(v,v') \in E_{G}$ and $v' \notin C$.
\end{itemize}
  Every finite play is a win for Cops, and every infinite play is a win
  for Thief. 
\end{definition}

The following Proposition provides a   useful variant of  entanglement games, see \cite[\S 6]{mathese}.

\begin{proposition} \label{modif:entag}
 Let $\ET{G,k}$ be the game played as the game $\Ent{G,k}$ apart
  that Cops is allowed to retire a number of cops placed on the
  graph. That is, Cops moves are of the form
\begin{itemize}
  \item  $(g,C,Cops) \rightarrow (g,C',Thief)$     (generalized skip move),
  \item  $(g,C,Cops) \rightarrow (g,C'\cup \set{g},Thief)$     (generalized replace move), 
\end{itemize} 
 where in both cases $C' \subseteq C$.
  Then Cops has a winning strategy in $\Ent{G,k}$ if and only if he
  has a winning strategy in $\ET{G,k}$.
\end{proposition}

We state one of the fundamental properties of the undirected graphs of bounded entanglement \cite[\S 6]{mathese}: 
\begin{theorem}\label{minorclosureTh}
 The class of undirected graphs of entanglement at most $n$,  for arbitrary fixed $n \in \mathbb{N}$, is minor closed, that is if $G$ is  minor of $H$ then $\Ent{G}\le \Ent{H}$. 
\end{theorem}
 On the one hand, atoms of Tutte's decomposition are cycles and $3$-connected components, on the other hand   cycles have entanglement at most $3$ \cite{BelkSanto0a7}. The aim now is to investigate  the structure of $3$-connected components whenever the starting graph has entanglement $3$. To this goal, we first establish   the relation between the three  notions of entanglement, connectivity and edge covering.

\begin{lemma} \label{Ent:connc:lemma1}
Let $G$ be a $k$-connected graph with $|V_G|\ge k+1$, then
  $$  k \le \Ent{G} \le \Cycl{G}$$
\end{lemma}
\begin{proof}
  Let us prove first the inequality $\Ent{G}\le \Cycl{G}$.  Let $X$ be an edge cover set of $G$, we shall show that Cops have a winning strategy in $\Ent{G,|X|}$. While Thief is moving on a path $\pi$, Cops strategy consists in placing a cop on a vertex $v$ of $\pi$ if and only if $v$ belongs to  $X$, showing that the number of cops placed on the graph is at most $|X|$.  If there is an infinite play then this implies that there is a cycle whose vertices are not covered by $X$, meaning that $X$ is not an edge cover set of $G$, this is a contradiction. Therefore, Cops's strategy is winning and hence $\Ent{G}\le \Cycl{G}$.  \\
  To prove the inequality $k\le \Ent{G}$, we use again  Menger's Theorem \ref{Menger:Theorem}, stating that  \emph{ a graph is $k$-connected if and only if it contains $k$ independent paths between any two vertices.}\\
  We shall prove that Thief has a winning strategy in
  $\Ent{G,k-1}$. To this goal  it is sufficient to show the following
  conditions to hold: whenever $k-1$ cops are placed on the graph, then
\begin{enumerate}[{(i)}]
\item  there is at least   an edge whose both  ends are not occupied by a cop. This is a consequence of the inequalities $\Cycl{G}\ge k >k-1$. The  inequality  $\Cycl{G}\ge k$ is justified by the fact that every edge cover set of the graph $G$  would disconnect it.       
\item  when it is Thief's turn to move from some vertex $v$,  then by Menger's Theorem,  it follows that  there is a free path (i.e. its vertices are not occupied by a cop) from $v$ to some vertex $w$ such that $ww'\in E_G$ and both $w$ and $w'$ are not occupied by a cop.
\end{enumerate}
It follows that Thief's  strategy consists in looking for an edge  which is not occupied by a  cop, choosing a free path to it,  use this path to reach this edge, and   iterating moves on it until Cops place a cop on one of its end points.  This strategy can be iterated infinitely often, hence its a winning strategy for Thief in  $\Ent{G,k-1}$.
 \end{proof}

\begin{lemma}\label{Ent:connc:lemma2}
Let $G$ be $k$-connected with  $\card{V_G}\ge k+1 $. If  $\Ent{G}=k$, then 
$$\Cycl{G}=\Ent{G}=k.$$
\end{lemma}
\begin{proof}
We have already mentioned that $\Cycl{G} \ge \Ent{G}=k$. To prove $\Cycl{G}=k$ we  assume that $\Cycl{G} > k$ and we shall deduce a contradiction: we shall prove that Thief has a winning strategy in the game $\Ent{G,k}$.\\
We distinguish two cases in the game $\Ent{G,k}$ according to the number of cops placed on the graph: 
\begin{proofbycases}
  \begin{caseinproof}\label{case1} If the  number of cops placed on the graph is at most $k-1$ then the same conditions \emph{(i)} and \emph{(ii)} provided in the proof of  Lemma \ref{Ent:connc:lemma1} still hold, hence in this case Thief will never be caught.
\end{caseinproof}
\begin{caseinproof} If $k$ cops are placed  on the graph then  consider the first position of the game for which the number of cops placed on the graph increases from $k-1$ to $k$, that is we consider the first  Cops' add move  
 $(v,C',Cops)\to (v,C'\cup \set{v}, Thief)$ where $\card{C'\cup\set{v}}=k$.\\
One the one hand,  since  $\Cycl{G}> k$  by assumption then there exists an edge $ww'$ which is not covered.
On the other hand, since a cop is posted on $v$ and $k-1$ cops are posted on the remaining vertices, then by Menger's Theorem  there are $k$ independent paths  linking $v$ to $w$, therefore there exists a free path $\pi$  from $v$ to $w$. Thief's  strategy consists in going from $v$ to $w$ through $\pi$ and iterating moves on the edge $ww'$ until Cops place a cop on either $w$ or $w'$, say $w$, giving rise to a position of the form $(w,C,Thief)$ where $|C|=k$ and $w\in C$, returning back  to the initial configuration. From the latter position Thief uses the same strategy described so far. Such a strategy  can be iterated infinitely often, hence Thief has a winning strategy in $\Ent{G,k}$. This contradicts the hypothesis  $\Ent{G}=k$.
\end{caseinproof}
\end{proofbycases}
\end{proof}

\section{The $k$-molecules}

In this section,  basing  on Lemma \ref{Ent:connc:lemma2},  we shall characterize the structure  of the  $k$-connected graphs of entanglement $k$, for arbitrary $k\ge 1$.  

\begin{definition}
 Let $k,h \ge 1$, $B_k=\set{b_1,\dots,b_k}$  and  $\mathcal{B} \subseteq B_k \times B_k$.  A $k$-molecule $\vartheta_{B_k}^{\mathcal{B},h}$  is the   graph $G=(V,E)$ such that  
\begin{enumerate}
\item  $V=B_k \cup \set{v_1,\dots,v_h}$, 
 \item  $E=\mathcal{B} \cup \set{ v_ib_j, \; 1\le i\le h ,\; 1\le j \le k}$. 
 \item $h \ge k-k'$,  where  $k'$ is the connectivity of the
   subgraph of $G$ induced by $B_k$.
\end{enumerate}
The set $B_k$ is called the base of the $k$-molecule.
\end{definition}
\begin{figure}[h]
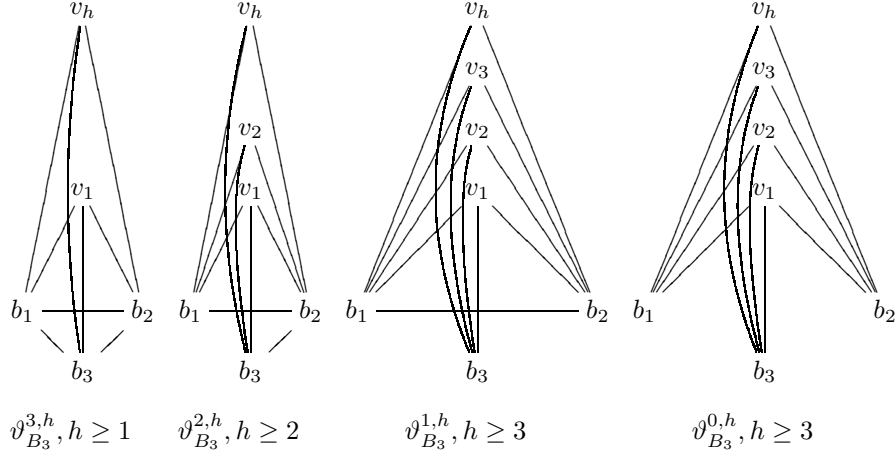

$$
  \xygraph{ %
  !{<0cm,0cm>;<0.8cm,0cm>:<0cm,0.8cm>::}
    []*+{b_1}="b1"
    [rr]*+{b_2}="b2"
    [dl]*+{b_3}="b3"
    [uuu]*+{v_1}="v1",
  "v1" [d(0)l(0.5)]="I1"
    "v1"[uuu]*+{v_h}="vh"
    "b1"-"b2"  "b2"-"b3" "b1"-"b3" "v1"-"b1" "v1"-"b2" "v1"-"b3" 
    "vh"-@{-}"b1"
     "vh"-@{-}"b2"
     "vh"-@{-}@`{"I1"}"b3"
 "b3"[d(1)l(0.2)]*+{\vartheta_{B_3}^{3,h}, h\ge 1}="index1"
   }\;
      \xygraph{ %
  !{<0cm,0cm>;<0.8cm,0cm>:<0cm,0.8cm>::}
    []*+{b_1}="b1"
    [rr]*+{b_2}="b2"
    [dl]*+{b_3}="b3"
    [uuu]*+{v_1}="v1"
     [u]*+{v_2}="v2"
  "v1" [d(0)l(0.5)]="I1"
  "v1" [d(0)l(0.85)]="Ih"  
 "v1"[uuu]*+{v_h}="vh"
    "b1"- "b2" "b2"-"b3"  "v1"-"b1" "v1"-"b2" "v1"-"b3" 
    "vh"-@{-}"b1"
     "vh"-@{-}"b2"
     "vh"-@{-}@`{"Ih"}"b3"
   "v2"-"b1"  "v2"-"b2" 
"v2"-@{-}@`{"I1"}"b3"
  "b3"[d(1)l(0.2)]*+{\vartheta_{B_3}^{2,h}, h\ge 2}="heo"
   }\;
   \xygraph{ %
  !{<0cm,0cm>;<0.8cm,0cm>:<0cm,0.8cm>::}
    []*+{b_1}="b1"
    [rrrr]*+{b_2}="b2"
    [dll]*+{b_3}="b3"
    [uuu]*+{v_1}="v1"
     [u]*+{v_2}="v2"
  "v1" [d(0)l(0.5)]="I1"
  "v1" [d(0)l(0.9)]="I2"  
  "v1" [d(0)l(1.4)]="Ih"
 "v1"[uu]*+{v_3}="v3"
 "v1"[uuu]*+{v_h}="vh"
    "b1"- "b2"  "v1"-"b1" "v1"-"b2" "v1"-"b3" 
    "vh"-@{-}"b1"
     "vh"-@{-}"b2"
     "vh"-@{-}@`{"Ih"}"b3"
   "v2"-"b1"  "v2"-"b2" 
"v2"-@{-}@`{"I1"}"b3"
"v3"-@{-}@`{"I2"}"b3"
"v3"-"b1" "v3"-"b2"
  "b3"[d(1)l(0.2)]*+{\vartheta_{B_3}^{1,h}, h\ge 3}="heo"
   }\;
   \xygraph{ %
  !{<0cm,0cm>;<0.8cm,0cm>:<0cm,0.8cm>::}
    []*+{b_1}="b1"
    [rrrr]*+{b_2}="b2"
    [dll]*+{b_3}="b3"
    [uuu]*+{v_1}="v1"
     [u]*+{v_2}="v2"
  "v1" [d(0)l(0.5)]="I1"
  "v1" [d(0)l(0.9)]="I2"  
  "v1" [d(0)l(1.35)]="Ih"
 "v1"[uu]*+{v_3}="v3"
 "v1"[uuu]*+{v_h}="vh"
 "v1"-"b1" "v1"-"b2" "v1"-"b3" 
    "vh"-@{-}"b1"
     "vh"-@{-}"b2"
     "vh"-@{-}@`{"Ih"}"b3"
   "v2"-"b1"  "v2"-"b2" 
"v2"-@{-}@`{"I1"}"b3"
"v3"-@{-}@`{"I2"}"b3"
"v3"-"b1" "v3"-"b2"
  "b3"[d(1)l(0.2)]*+{\vartheta_{B_3}^{0,h}, h\ge 3}="heo"
   }\;
$$
\caption{The structure of $3$-molecules $\vartheta_{B_3}^{\card{\mathcal{B}},h}$.}
\label{kmoleculesFig}
\end{figure}
The $1$-molecules are the \emph{stars}, the $2$-molecules have  been discussed in \cite{BelkSanto0a7}, and the  $3$-molecules $\vartheta_{B_3}^{\mathcal{B},h}$ are pictured (up to graph isomorphism) in Figure \ref{kmoleculesFig}; since no confusion will arise we have substituted  $\mathcal{B}$ by its cardinality $\card{\mathcal{B}}$\footnote{This is possible without confusion when   $k=3$, because a graph on $3$ vertices is completely determined (up to isomorphism) by the number of its edges. }.   

 The following Definition formalizes the fact that  a graph is  a $k$-molecule.

\begin{definition}
Let $G$ be a graph and $B\subset V_G$ with $|B|=k$.  We say that the pair  $(G,A)$ is a $k$-premolecule if there exists a   $k$-molecule $\vartheta_{B_k}^{\mathcal{B},h}$  and a graph isomorphism $\psi: G \rTo \vartheta_{B_k} ^{\mathcal{B},h}$ sending $B$ to $B_k$. \\
We say that a $G$ is an abstract $k$-molecule if there exists $B\subset V_G$ with $|B|=k$ such that the pair $(G,B)$ is a $k$-premolecule.
 \end{definition}

\begin{lemma}\label{DiscretAk:Lemma}
Let  $G$ be  an abstract $k$-molecule. If there exist  $A_k,B_k \subset V_G$  such that each pair $(G,A_k)$ and $(G,B_k)$ is a $k$-premolecule. Then, 
\begin{enumerate}[{4.}i] 
\item   the three subgraphs of $G$  induced by $A_k\setminus B_k$, $B_k \setminus A_k$ and $V_G\setminus(A_k \cup B_k) $ are all discrete,
\item  the subgraphs of $G$ induced by ${A}_k$ and ${B}_k$ are isomorphic.
\end{enumerate}
\end{lemma}
\begin{proof} 
$\textrm{ }$
 \begin{enumerate}[{4.}i]
\item Assume that $A_k\setminus B_k$ is not discrete i.e. there exists $a_1,a_1 \in A_k$ with $a_1a_2 \in E_{G_1}$. If we consider the $k$-premolecule $(G,B_k)$ then we observe that $a_1a_2\in V_{G}\setminus B_k$, this is a contradiction since the graph induced by $V_{G}\setminus B_k$ must be discrete by the definition of the $k$-molecules. We get that the sugraph induced by $B_k \setminus A_k$ is discrete by just  dualinzing   the above proof. To argue that  $V_G \setminus (A_k \cup B_k)$ is discrete observe that $V_G \setminus (A_k \cup B_k) \subseteq V_{G}\setminus  A_k$, and since $V_{G}\setminus  A_k$ is discrete by the definition of the $k$-molecules, then it follows that $V_G \setminus (A_k \cup B_k)$ is also discrete. 
\item  Consider the  mapping $\psi: {G} \rTo G$ such that $\psi$ is a bijection from    $A_k \setminus B_k$ to $B_k \setminus A_k$,  and it is the identity on both  $(A_k \cap B_k)$ and $V_G\setminus (A_k \cup B_k)$. Since  the subgraphs induced by $A_k \setminus B_k, B_k\setminus A_K$, and $V_G \setminus (A_K \cup B_k)$ are discrete by 4.i, then  for all $v_1,v_2 \in A_k$ and for all $w_1,w_2 \in B_k$,  $\psi(v_i)=\psi(w_i), i=1,2,$ if and only if $\psi(v_1v_2)=\psi(w_1w_2)$. Hence the subgraphs of $G$ induced by $A_k$ and $B_k$ are isomorphic.
\end{enumerate}
\end{proof}

\cutout{\begin{corollary}\label{From:FVS:To:Kmolecule:Cor}
Let $(G,B)$ be a $k$-premolecule and $X\subset V_G$ a minimal edge cover of $G$ then the pair $(G,X)$ is also a $k$-premolecule
\end{corollary}
}
In order to compute the connectivity of the $k$-molecules, the following Lemma provides a construction of them and a upper bound of their connectivity. 
\begin{lemma}\label{ConstructConnect:Lemma}
Let $G$ be the graph constructed as follows: out of a  graph $B$ and a set of vertices $\set{v_1,\dots,v_h}$ add an edge between each $v_i,i=1,\dots,h$ and each $b$ in  $B$. If $B$ is $k'$-connected,  then the connectivity of $G$ is at least  $$min(|V_B|,k'+h).$$
\end{lemma}
\begin{proof}
Let $m=min(|V_B|,k'+h)$,  we shall prove that $G$ is $m$-connected, by  proving that every two vertices $x,y$ are linked  by at least $m$ disjoint paths. We split the proof in three cases.
\begin{proofbycases}
\begin{caseinproof} If $x,y \in V_B$, then there are $k'$ disjoint path in the subgraph induced  by $V_B$ from $x$ to $y$ because the latter is $k'$-connected. Moreover there are  $h$ disjoint paths of the form $xv_1y,x_2y,\dots xv_{h}y$ where $v_i \in V_G \setminus V_B$.  
\end{caseinproof}
\begin{caseinproof} If $x,y \in V_G \setminus V_B$. In this case there are  $|V_B|$ disjoint paths of the form  $xb_1y, xb_2y,\dots, xb_ky$, where $b_i \in V_B$.  
\end{caseinproof}
\begin{caseinproof} If $x\in V_B$ and $y \in V_G\setminus V_B$, then the  $k'+h$ disjoint paths are $\Pi_1 \cup \Pi_2 \cup \Pi_3$ where :
\begin{itemize}
\item $\Pi_1=\set{xy}$, recall that $xy\in E_G$ by definition.
\item To exhibit $\Pi_2$ recall first that since the graph $B$ is $k'$ connected then $x$ has at least $k'$ neighbors in $B$, let $b_1,\dots, b_{k'}$ be such neighbors. Therefore we let 
$ \Pi_2=\set{xb_1y,\dots,xb_{k'}y} $.
\item Finally,
$$\Pi_3=\set{xv_1b'_{1}y,\dots,xv_ib'_{i}y,\dots xv_{h-1}b'_{h-1}y}$$
where $\set{b_1,\dots,b_{k'}} \cap \set{b'_1,\dots,b'_{h-1}}=\emptyset$.\\
It is straightforward to check that paths in $\Pi_1\cup \Pi_2 \cup \Pi_3$ are disjoint, they share just their two end points. Moreover  $\card{\Pi_1\cup \Pi_2 \cup \Pi_3}=1 + k' + (h-1)=k'+h$.  
\end{itemize}
\end{caseinproof}
\end{proofbycases}
\end{proof}

Now we state the main combinatorial properties  of  the $k$-molecules.

\begin{proposition}
 Let $G=\vartheta_{B_k}^{\mathcal{B},h} $ be a $k$-molecule. Then
\begin{enumerate} 
 \item  the connectivity of $G$ is  $k$, 
 \item $G$ has $B_k$ as a minimal edge cover  and hence $\Cycl{G}=k$, and 
 \item the entanglement of $G$ equals  $k$.
\end{enumerate}
\end{proposition}
\begin{proof}
$\textrm{ }$
\begin{enumerate}
\item Observe that if  $G$  is a clique  then it should be   a $(k+1)$-clique. Moreover, this holds  if and only if the subgraph of $G$  induced by $\mathcal{B}_k$ is a $k$-clique and $h=1$.  In this case the connectivity of $G$ is $k$ by Convention \ref{Conn:Clique:Conv}.   Assume that $G$ is not a clique. On the one hand,   by lemma  \ref{ConstructConnect:Lemma}, it follows that  the connectivity of $G$ is at least  $min(|B_k|,k'+h)=min(k,k'+h)$ and since $h \ge k-k'$ by the definition of the $k$-molecules, then $min(k,k'+h)=k$, showing that the connectivity of $G$ is at least $k$.  On the other hand,  we need the Claim:
\begin{claim}
Let $G:=\vartheta_{B_k}^{\mathcal{B},h}$ be a $k$-molecule. If $G$ is not a clique then $h\ge 2$.
\end{claim}
\begin{proof}
We distinguish two cases according to the nature of $G[B_k]$. \\
If $G[B_k]$ is a $k$-clique, then we need $h\ge 2$, because if $h=1$ then $G$ would be a $(k+1)$-clique contradicting the hypothesis.\\
If $G[B_k]$ is not a clique, then there exist  at least two vertices $b,b' \in B_k$ such that $bb'\notin E_G$, therefore $B_k \setminus \set{b,b'}$ is a $(k-2)$-separator of $b$ from $b'$ in $G[B_k]$. This implies that the connectivity of $G[B_k]$ is at most $k-2$. From the definition of the $k$-molecules, we have $h \ge k-k'$ where $k'$ is the connectivity of $G[B_k]$, hence $h \ge k - (k-2)=2$.
This ends the proof of the Claim.
\end{proof}
Since $G$ is not a clique then, according to the Claim, we have  $h\ge 2$ and hence $B_k$ is a $k$-separator in $G$ of  any two vertices in $V_G \setminus B_k$ and in this case $G$ can not be $(k+1)$-connected.  We conclude that  the connectivity of $G$  is $k$.
\item Since $G$ is $k$-connected and $|V_G| \ge k+1$ then by Lemma \ref{Ent:connc:lemma1} we get $\Cycl{G} \ge k$. It is easy to check that $B_k$ is an edge cover of $G$, hence $\Cycl{G}\le \card{B_k} =k$. Therefore $\Cycl{G}=k$.
\item On the one hand, since $G$ is $k$-connected then from  Lemma \ref{Ent:connc:lemma1}  we obtain  $\Ent{G}\ge k$.  On the other hand, from the same Lemma \ref{Ent:connc:lemma1}, we have $\Ent{G}\le \Cycl{G}$ and from the previous item we got $\Cycl{G}=k$, thus $\Ent{G} \le k$. We conclude that $\Ent{G}=k$. 
\end{enumerate}
\end{proof}

So far we have stated and proved the main properties of the $k$-molecules. Conversely, the following Proposition characterizes the  graphs for which the entanglement, the connectivity and the  cyclicity coincide.  

\begin{proposition}  \label{3Conn:3Entang:3Mol:Lemma} 
If $G$ is   $k$-connected  with $\Ent{G}=k$, then $G$ is an abstract $k$-molecule.
\end{proposition}
\begin{proof}
Let $G$ be a graph as stated.  Since $\Ent{G}=k$ and $G$ is $k$-connected then by  Lemma \ref{Ent:connc:lemma2} it follows  that $\Ent{G}=\Cycl{G}=k$, thus let $B_k=\set{b_1,\dots,b_k}$ be a minimal edge cover set  of $G$ and let $k'$ be the connectivity of $G[B_k]$.  \newline
We claim that vertices in $V_G\setminus B_k$ are at distance one from $B_k$ because $B_k$ is an edge cover set of $G$. For the same reason  the subgraph  of $G$ induced by $V_G\setminus B_k$ is discrete. Thus, for each $v\in V_G \setminus B_k$ there exists $b\in B_k$ such that $vb \in E_G$.  Let $v\in V_G\setminus B_k$ and $\mathcal{N}_v$ be
the set of its neighbors, clearly  $\mathcal{N}_v \subseteq B_k$ because $G[V_G\setminus B_k]$ is discrete. If  $\card{\mathcal{N}_v}< k$  meaning that $\mathcal{N}_v \subsetneq B_k$ then clearly $\mathcal{N}_v$ separates $v$ from $B_k\setminus \mathcal{N}_v$. This contradicts the assumption that $G$ is $k$-connected. We conclude that $\mathcal{N}_v=B_k$ for each $v\in V_G \setminus B_k$.\\
Finally, to accomplish the proof that $G$, coming with the desired data,  is a $k$-molecule, it remains  just to show that $\card{V_G\setminus B_k}\ge k-k'$ where $k'$ is the connectivity of $G[B_k]$.  Towards a contradiction, assume that   $\card{V_G\setminus B_k}< k-k'$. Since the connectivity of $G[B_k]$ is $k'$ then there exists a $k'$-separator  in $G[B_k]$, let $S_{k'}$ be such a  separator and assume that it separates $b_1$ from $b_2$; $b_1,b_2 \in B_k$.  Therefore $S_{k'} \cup (V_G\setminus B_k)$  separates  also $b_1$ from $b_2$.   A simple computation shows that  $\card{S_{k'} \cup (V_G\setminus B_k)}< k$, contradicting the fact that $G$ is $k$-connected. \\
We conclude that the pair  $(G,B_k)$ is a $k$-premolecule, moreover $G$ may be written  as  $\vartheta_{B_k}^{\mathcal{B},h}$ where $\mathcal{B}$ are the edges of the subgraph of $G$ induced by $B_k$, and $h=\card{V_G\setminus B_k}$.  
 \end{proof}

\subsection{Classification of the $3$-molecules}
A natural question arises:  in how  many manners a $k$-molecule may be written? To illustrate this question consider the graph $G$ depicted in Figure \ref{2premolecule}. If we take $A_2=\set{a,b}$, then  the graph $G$ may be viewed as the two molecule $\vartheta_{A_2}^{\emptyset,2}$. If we take $B_2=\set{c,d}$, then $G$ may be considered  as the $2$-molecule $\vartheta_{B_2}^{\emptyset,2}$. \cutout{It seems that the number of the manner by which a $k$-molecule may be written depends on the number of its possible  bases.   }
\begin{figure}
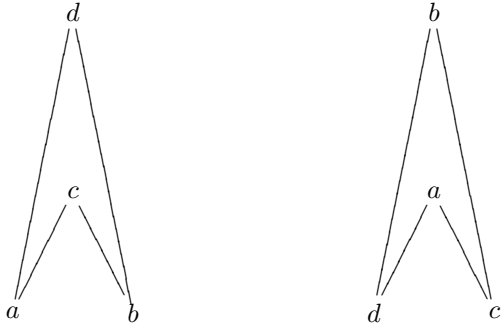

$$  \xygraph{ %
  !{<0cm,0cm>;<0.8cm,0cm>:<0cm,0.8cm>::}
 []*+{a}="a"
 [rr]*+{b}="b"
 [luu]*+{c}="c"
 [uuu]*+{d}="d"
"c"-"a"
"c"-"b"
"d"-"a"
"d"-"b"
\;\;\;
 [rrrr]*+{d}="d2"
 [rr]*+{c}="c2"
 [luu]*+{a}="a2"
 [uuu]*+{b}="b2"
"c2"-"a2"
"c2"-"b2"
"d2"-"a2"
"d2"-"b2"
}
$$
\caption{Two possible ways to view a $2$-molecule.}
\label{2premolecule}
\end{figure}

A  $k$-molecule may admit many bases giving rise to what we  call \emph{ambiguous}  molecules. The formal definition of ambiguity follows.
\begin{definition}
Let $G$ be an abstract  $k$-molecule. We say that $G$ is  non  ambiguous if there exists just one set $B\subset V_G$ such that the pair $(G,B)$  is a $k$-premolecule.  
Similarly, a $k$-molecule $\vartheta$ is non ambiguous if $\vartheta$ viewed as a graph is  non ambiguous. 
\end{definition}

 The following  Proposition gives  an explicit  characterization    the class of  ambiguous  $3$-molecules.  
\begin{proposition}\label{Ambiguous:Prop}
Let $\vartheta$ be a $3$ molecules $\vartheta_{B_3}^{\mathcal{B},h}$, then $\vartheta$  is ambiguous if and only if one of the following cases holds: 
\begin{enumerate}[I.]
\item $|\mathcal{B}|=3$ and $h=1$, that is $\vartheta$ is the $4$-clique,
\item $|\mathcal{B}|=2$ and $h=2$,
\item $|\mathcal{B}|=0$ and $h=3$, that is $\vartheta$ is  the complete bipartite graph $K_{3,3}$.
 \end{enumerate}
These graphs are depicted in figure \ref{Ambiguous:Fig}
\begin{figure}[h]
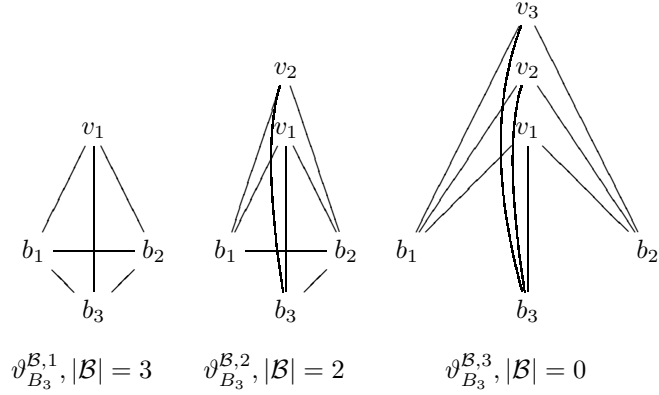

$$
  \xygraph{ %
  !{<0cm,0cm>;<0.8cm,0cm>:<0cm,0.8cm>::}
    []*+{b_1}="b1"
    [rr]*+{b_2}="b2"
    [dl]*+{b_3}="b3"
    [uuu]*+{v_1}="v1",
  "v1" [d(0)l(0.5)]="I1"
    "b1"-"b2"  "b2"-"b3" "b1"-"b3" "v1"-"b1" "v1"-"b2" "v1"-"b3" 
 "b3"[d(1)l(0.2)]*+{\vartheta_{B_3}^{\mathcal{B},1}, |\mathcal{B}|=3}="index1"
   }\;\;\;
      \xygraph{ %
  !{<0cm,0cm>;<0.8cm,0cm>:<0cm,0.8cm>::}
    []*+{b_1}="b1"
    [rr]*+{b_2}="b2"
    [dl]*+{b_3}="b3"
    [uuu]*+{v_1}="v1"
     [u]*+{v_2}="v2"
  "v1" [d(0)l(0.5)]="I1"
  "v1" [d(0)l(0.85)]="Ih"  
    "b1"- "b2" "b2"-"b3"  "v1"-"b1" "v1"-"b2" "v1"-"b3" 
   "v2"-"b1"  "v2"-"b2" 
"v2"-@{-}@`{"I1"}"b3"
  "b3"[d(1)l(0.2)]*+{\vartheta_{B_3}^{\mathcal{B},2}, |\mathcal{B}|=2}="heo"
   }\;\;\;
\xygraph{ %
  !{<0cm,0cm>;<0.8cm,0cm>:<0cm,0.8cm>::}
    []*+{b_1}="b1"
    [rrrr]*+{b_2}="b2"
    [dll]*+{b_3}="b3"
    [uuu]*+{v_1}="v1"
     [u]*+{v_2}="v2"
  "v1" [d(0)l(0.5)]="I1"
  "v1" [d(0)l(0.9)]="I2"  
  "v1" [d(0)l(1.35)]="Ih"
 "v1"[uu]*+{v_3}="v3"
 "v1"-"b1" "v1"-"b2" "v1"-"b3" 
   "v2"-"b1"  "v2"-"b2" 
"v2"-@{-}@`{"I1"}"b3"
"v3"-@{-}@`{"I2"}"b3"
"v3"-"b1" "v3"-"b2"
  "b3"[d(1)l(0.2)]*+{\vartheta_{B_3}^{\mathcal{B},3}, |\mathcal{B}|=0}="heo"
   }\;
$$
\caption{The set of ambiguous  $3$-molecules.}
\label{Ambiguous:Fig}
\end{figure}
\end{proposition}  
\begin{proof}
First, we  give a useful property of ambiguous $k$-molecules:
\begin{claim}
If $G$ is  an ambiguous $k$-molecule, and $A_k,B_k$ are two different bases of $G$, then $A_k \cup B_k = V_{G}$. 
\end{claim}
\begin{proof}
Assume that $A_k \cup B_k \subsetneq V_{G}$ and then let $w\in V_{G} \setminus (A_k \cup B_k)$. Hence, for each $b \in B_k$ we have by the definition of the $k$-molecules that $wb \in E_{G}$. This implies that the subgraph of $G$ induced by $V_{G} \setminus A_k$ is not discrete, contradicting the definition of the $k$-molecules. We conclude that $A_k\cup B_k =V_{G}$. 
This ends the proof of Claim.
\end{proof}
Let $G:=\vartheta_{B_3}^{\mathcal{B},h}$ be an ambiguous $3$-molecule and $A_3,B_3$ be two distinct bases of $G$, thus by the previous Claim  we get  $A_3\cup B_3 =V_{G}$.  
We recall first that Lemma  \ref{DiscretAk:Lemma} states that both the graphs induced by $A_3 \setminus B_3$ and $B_3 \setminus A_3$ are discrete and moreover the graphs induced by $A_3$  and $B_3$  are isomorphic.  We distinguish three cases according to $|A_3 \cap B_3|$. 
\begin{enumerate}
\item $|A_3\cap B_3|=0$. Since both $G[A_3]$ and $G[B_3]$ are discrete,  then $G$ is the complete bipartite graph $K_{3,3}$, i.e. $G$ is the $3$-molecule $\vartheta_{B_3}^{\mathcal{B},3}$ where $\mathcal{B}=0$.
\item  $|A_3\cap B_3|=1$. Let $A_3=\set{w,a_1,a_2}$ and $B_3=\set{w,b_1,b_2}$. First, $a_1a_2\notin E_G$ because $G[A_3 \setminus B_3]$ is discrete, and also $b_1b_2 \notin E_G$ because $G[B_3\setminus A_3]$ is discrete.  Second, we shall argue that  $\set{wa_i,wb_i,\; i=1,2} \subset E_G$. Assume that $wa_1 \notin E_G$, then  $\set{b_1,b_2}$  is a $2$-separator in $G$ of $a_1$ from $a_2$. This is a contradiction because the $3$-molecules are $3$-connected. We deduce that $wa_1\in E_G$. By symmetry, we obtain also that $wa_2,wb_1,wb_2 \in E_G$. We conclude that,  in this case, $G$ is the $3$-molecule $\vartheta_{B_3}^{\mathcal{B},2}$ where $\mathcal{B}=2$.   
\item  $|A_3\cap B_3|=2$.  Observe that in this case $h=1$ in $\vartheta_{B_3}^{\mathcal{B},h}$. We shall argue next that both $G[A_3]$ and $G[B_3]$  are the $3$-clique. Recall that in the $k$-molecule $\vartheta_{B_k}^{\mathcal{B},h}$  we have $h \ge k-k'$ where $k'$ is the connectivity of $G[B_k]$. If  $|E_{G[A_3]} |\le 2$, then the connectivity of $G[A_3]$ is at most $1$ and hence in the $3$-molecule $\vartheta_{A_3}^{\mathcal{B},h}$  we should have $h \ge k-k' \ge 3-1=2$, this is a contradiction because we have already mentioned that $h=1$. We conclude that $|E_{G[A_3]}|=3$ meaning that $G$ is the $4$-clique, i.e. the $3$-molecule $\vartheta_{B_3}^{\mathcal{B},1}$  where $|\mathcal{B}|=3$.  
\end{enumerate}
\end{proof}

\section{Tree decomposition of graphs of entanglement at most $3$}
The main result of  \cite{BelkSanto0a7} states  that a graph of entanglement $2$ has a tree decomposition into $2$-connected components   such that the latter are the $2$-molecules that come with a prescribed set of articulation points. Conversely, starting with the $2$-molecules and the $1$-Sum operator we have been able to generate the class of graphs of entanglement $2$. Now we shall follow this approach to deal with  the class of graphs of entanglement $3$.   \\ 
 We find some \emph{necessary}  conditions on the structure of the Tutte's tree to be a tree decomposition of a $2$-connected graph of entanglement $3$. The necessary conditions deal with three features of the tree: \emph{(i)} conditions on the structure of the $3$-connected components: they are the $3$-molecules. This is  direct consequence of the results proved in the previous sections, \emph{(ii)} conditions on the hinges are given in a similar way of those given  on the articulation points  when the starting graph has entanglement $2$,  and \emph{(iii)} conditions on the diameter of the tree.

\subsection{Necessary conditions on Tutte's tree }
Since a $2$-connected  graph may be written by means of  the $2$-Sum operator, we begin by inspecting the main cases for which the $2$-Sum operator increases the entanglement \footnote{This idea has been already considered in  \cite{BelkSanto0a7} where we have looked for the cases for which the $1$-sum operator increases the entanglement.  There, the  $1$-sum operator which does not increase the entanglement is called the legal $1$-Sum.}. In the next, the symbol  $+$ denotes the $2$-Sum operator on $2$-connected graphs  given in Definition \ref{twosum:def}.\\
From now on   we shall  deal particularly with the $3$-molecules and write  $\vartheta_{a,b,c}^{\mathcal{B},h}$ instead of $\vartheta_{\set{a,b,c}}^{\mathcal{B},h}$. 

\begin{lemma}\label{2SumOn3Connect:Lemma2}
Let   $\vartheta_{a,b,c}^{\mathcal{B},h}$ be  a $3$-molecule and   $C_3$ be the $3$-cycle on the vertices $\set{a,v_1,z}$.  Define the graph $G$ as follows:
$$
G=C_3 +_{av_1} \vartheta_{a,b,c}^{\mathcal{B},h}.
$$
If   $\set{a,v_1}$ does not belong to any  minimal  edge cover set   of $\vartheta_{a,b,c}^{\mathcal{B},h}$, then $\Ent{G} \ge 4$
\end{lemma}
\begin{proof}
The graph $G$ may be viewed as the graph that results  from inserting the new vertex $z$ in the edge $av_1$ of $\vartheta_{a,b,c}^{\mathcal{B},h}$. Let us abbreviate $\vartheta_{a,b,c}^{\mathcal{B},h}$ by $\vartheta$. \\
Define $f: G \rTo \vartheta $ as follows: $f(v)=v$ if $v\neq z$, and $f(z)=a$.  Thief's strategy in $\ET{\vartheta,3}$\footnote{The game $\ET{G,k}$ is defined as the game $\Ent{G,k}$ apart that Cops can retire a number of cops, the two versions are equivalent, Proposition \ref{modif:entag}.} ( even if  it is losing) that forces Cops to occupy a minimal edge cover set of $\vartheta$, this strategy exists by Lemma \ref{OnlyWinningStratMolleLemma}, will allow us to construct a winning strategy for Thief  in $\Ent{G,3}$. Every position $(v,C_{\vartheta},P)$ of $\ET{\vartheta,3}$  is matched with a position $(g,C_{G},P)$ of $\Ent{G,3}$ such that the following conditions hold:
\begin{align*}
& \bullet  f(g)=w \tand f(C_G)=C_{\vartheta}. \label{COPS2} \tag{COPS} \\
& \bullet   \textrm{ if Thief moves from } (a,C_{\vartheta},Thief) \textrm{ to } (v_1,C_{\vartheta},Cops) ,\\ & \hspace{0.5 cm} \textrm{ then }  a \in C_{\vartheta}. \label{THIEF-Z2} \tag{COP-ON-$a$}
\end{align*}
The condition (\ref{COPS2}) states essentially that the  cops in $\ET{\vartheta,3}$ are  placed on the image of cops in $\Ent{G,3}$ by the function $f$ defined above. The condition (\ref{THIEF-Z2})   states that whenever Thief leaves vertex $a$ to $v_1$ then a cop  must   already be   placed on $a$.\\
A Thief's move $M_{\vartheta}=(u,C_{\vartheta},Thief)\to (w,C_{\vartheta},Cops)$ in  $\ET{\vartheta,3}$ is simulated either by a move or a sequence of moves in $\Ent{G,k}$ according to the edge $uw$.
\begin{enumerate}
\item If $uw\neq v_1a$, then the move $M_G$ is simulated in $\Ent{G,3}$ by the same Thief's move. Observe that the latter move is possible in $\Ent{G,3}$ because $ w \in C_G$ then already $f(w) =w \in f(C_G)=C_{\vartheta}$, and this is impossible. 

\item If $(u,w)=(v_1,a)$, then the move $M_G$ is simulated in $\Ent{G,3}$ by the following sequence of moves:
\begin{align*}
(v_1,C_{G},Thief) \to (z,C_G,Cops) \to (z,C'_G,Thief) \to (a,C'_G,Cops)
\end{align*}
This sequence is possible. If Thief can not perform such moves then either $z\in C_G$ or $a \in C'_G$. If $z\in C_G$ then $f(z)=a \in f(C_G)=C_{\vartheta}$, which  is impossible. If $a \in C'_G$, then already $a \in C_G$ and hence $f(a)=a \in f(C_G)=C_{\vartheta}$, which is also impossible.
Observe that in this case,   Thief's move in $\ET{\vartheta,3}$ is simulated by a sequence of Thief's  moves in $\Ent{G,3}$, and the latter are interleaved with  Cops moves, and then the position of Cops in $\ET{\vartheta,3}$ such be updated using the function $f$.  So, it remains to show that Cops related moves in $\ET{\vartheta,3}$ respect the rule of the game. \\
From the latter position in  $\Ent{G,3}$, Cops  next move is a  from $(a,C'_G,Cops) \to (a,C''_G,Thief)$. Hence in $\ET{\vartheta,3}$ Cops related move should be of the form $(a,C_{\vartheta},Cops) \to (a,f(C''_G),Thief)$. Let us compute $f(C''_G)$ in term of $C_{\vartheta}$. Observe first that $C''_{G} =(C_G \setminus  A) \cup B$ where $\emptyset \subseteq  B \subseteq  \set{z,a}$ and $A\subset C_G$ with $|A|\le 2$.
\begin{align*}
f(C''_G) &= f [(C_G \setminus A) \cup B \\
            &=  f(C_G \setminus A) \cup f(B) \\
            &= [f(C_G) \setminus f(A)) \cup Z ] \cup f(B) 
\end{align*}
where $\emptyset \subseteq Z \subseteq \set{z}$  and $f(B)\subseteq \set{z,a}$.
Therefore,
\begin{align*}
f(C''_G)&=(f(C_G) \setminus f(A)) \cup (Z \cup f(B)) \\
         & = (C_{\vartheta}\setminus f(A)) \cup Z'
\end{align*}
On the one hand  $A\subset C_G$ and hence  $f(A)\subset f(C_G)=C_{\vartheta}$. On the other hand  $Z'=Z \cup f(B) \subseteq \set{z} \cup \set{z,a}=\set{z,a}$. We conclude that Cops'  move under discussion respects the rules of the game. 
\item If $(u,w)=(a,v_1)$, then the move $M_{\vartheta}$ is simulated in $\Ent{G,3}$ by Thuief's iteration on $az$ until a cop is placed either on $a$ or $z$ and then Thief goes to $v_1$. That is, it is the following sequence:
\begin{align*}
M^{\star}_{G}=(a,C_G,Thief)\to (z,C_G,Cops) & \to (z,C_G,Thief) \to (a,C_G,Cops)\\
 & \to (a,C_G,Thief)\to (z,C_G,Cops) \\
 & \to  \dots  \\
 &\to (a,C_G,Cops) \to (a,C'_G,Thief)  \\ 
& \to (z,C'_G,Cops) \to (z,C''_G,Thief) \\
 & \to (v_1,C''_G,Cops) 
\end{align*}
Such that $C'_G\neq C_G$ or $C''_G\neq C_G$. Let us show that Thief can perform such moves, i.e.  $z \notin C_G$ and $v_1 \notin C''_G$. \\
If $v_1 \in C''_G$ then $v_1 \in C_G$, and hence $f(v_1)=v_1 \in f(C_G)=C_{\vartheta}$,  this is impossible.  \\
If $z \in C_G$ the $f(z) = a \in f(C_G)=C_{\vartheta}$, then let us come back to  the previous round of simulation.  We mean if we consider Thief's previous moves in $\ET{\vartheta,3}$, then they are of the form 
$$
(a^{-1},C_{\vartheta}^{-1},Thief ) \to (a,C_{\vartheta}^{-1},Cops) \to (a,C_{\vartheta},Thief)
$$
and since  we have  supposed  that  $z \in C_G$, then  $f(z)=a \in f(C_G)=C_{\vartheta}$ and moreover $a \in C_{\vartheta}^{-1}$, which is impossible.  \\
Let us argue now that Cops' next move in $\ET{\vartheta,3}$ respects   the rules of the game. From the latter position in $\Ent{G,3}$, Cops' next move would be of the form $(v_1,C''_G,Cops) \to (v_1,C_{G}^{\star},Thief)$, and hence Cops' in $\ET{\vartheta,3}$ would reply, according to condition (\ref{COPS2}) by the move $(v_1,C_{\vartheta},Cops) \to (v_1,f(C_{G}^{\star}),Thief)$. Observe first that $a \in C_{\vartheta}$  by the condition (\ref{THIEF-Z2}). In order to show that the latter move respect the rules of the game, we compute $f(C_{G}^{\star})$ in term of $C_{\vartheta}$.  Note that  $C_{G}^{\star}=(C_{G} \setminus A ) \cup B$ where  $\emptyset \subseteq B \subseteq \set{a,z,v_1}$ and $A \subseteq C_G$. As in the previous case, a simple computation shows that \\
$ f(C_G^{\star}) = (C_{\vartheta} \setminus f(A)) \cup f(B')$ where $B' \subseteq \set{a,z,v_1}$. On the one hand $A\subseteq C_{G} $ implying $f(A) \subseteq f(C_{G})=C_{\vartheta}$. On the other hand $f(B') \subseteq \set{a,v_1}$. However, we have mentioned that $a \in C_{\vartheta}$, and hence $f(C_G^{\star})=(C_{\vartheta} \setminus f(A))  \cup B''$  where $B'' \subseteq \set{v_1}$. This shows that Cops' move in question respects the rules of the game.  
\end{enumerate}

So far we have described the simulation between the games $\Ent{G,3}$ and $\ET{\vartheta,3}$ and shown that it is consistent. Now  we shall show that the hypothesis of the  Lemma under proof imply  implicitly some restrictions on  the $3$-molecule $\vartheta_{a,b,c}^{\mathcal{B},h}$ provided in this Lemma.  
\begin{claim}
the $3$-molecule $\vartheta_{a,b,c}^{\mathcal{B},h}$ described in Lemma \ref{2SumOn3Connect:Lemma2} is not the $4$-clique, and hence $h\ge 2$.
\end{claim} 
\begin{proof}
If  $\vartheta:=\vartheta_{a,b,c}^{\mathcal{B},h}$  is the $4$-clique, then any set of vertices of size $3$ forms a minimal edge cover of $\vartheta$, implying that   ${a,v_1}$ belongs to some minimal edge cover of $\vartheta$, contradicting the hypothesis.  This ends the proof of the Claim.
\end{proof}
Assume that $\vartheta$ is ambiguous and let   $\set{a',b',c'}$ be an other minimal edge cover of $\vartheta$. Since the edge $v_1a$, where the $2$-Sum is performed, does not belong to any minimal edge cover of $\vartheta$, then it is invariant w.r.t. changing   the bases of $\vartheta$. This shows that we can deal with $\vartheta$ like wise $\set{a,b,c}$ is the unique base.

If Thief is trapped in $\ET{\vartheta,3}$ then we have a position of the form \\ $(v,\set{a,b,c},Thief)$ where $v\in \set{a,b,c}$. The latter position is matched with a position $(v,C_G,Thief)$ of $\Ent{G,3}$ where  $f(C_G)=\set{a,b,c}$. Therefore, either $C_G=\set{a,b,c}$ or $C_G=\set{z,b,c}$.
\begin{proofbycases}
\begin{caseinproof}If $C_G=\set{a,b,c}$, then Thief can go to $v_1$ an iterates moves on $v_1z$ forcing Cops to put a cop on $v_1$ or $z$. If Cops put a cop on $v_1$ then the image of Cops on $\vartheta$ by  $f$ is no longer a minimal edge cover, and hence Thief plays in $\ET{\vartheta,3}$  with the strategy   that consists in forcing  Cops to occupy a minimal edge cover set of $\vartheta$, i.e. $\set{a,b,c}$. If Cops put a cop on $z$, then this cop  comes from  $a,b,$ or $c$. \\
If this cop comes from $b$ or $c$, then the image of cops on $\vartheta$ by $f$ is either $\set{a,c}$ or $\set{a,b}$ which is not a  minimal edge cover set of $\vartheta$, and hence Thief forces Cops to occupy again a minimal edge cover of $\vartheta$. \\
If this cop comes from $a$ then we go    to \emph{Case(ii)}.
\end{caseinproof}
\begin{caseinproof} If $C_G=\set{z,b,c}$, then assume that the current vertex occupied by Thief is denoted by $x \in \set{z,b,c}$, and hence Thief's can choose the path $xav_2$ and\footnote{The vertex $v_2$ exists because $h\ge 2$ in   $\vartheta=\vartheta_{_{a,b,c}^{\mathcal{B},h}}$}  iterates moves on $v_2a$ forcing Cops to put a cop on $a$,  coming back to \emph{Case(i)}. 
\end{caseinproof}
\end{proofbycases}
Such a strategy for Thief in $\Ent{G,3}$ can be iterated infinitely often, that is, it is a winning strategy, and hence $\Ent{G}\ge 4$. This ends the proof of Lemma \ref{2SumOn3Connect:Lemma2}.
\end{proof}

\begin{corollary}\label{Interface:Cor}
Let $G$ be $2$-connected such that $\Ent{G}=3$ and $T$ be its  Tutte decomposition. Let $t_1t_2 \in E_T$ such that the torso $\tau_{t_1}$ is $3$-connected, then there exists $B_3 \subset V_{\tau_{t_1}}$ such that  \emph{(i)} the pair $(\tau_{t_1},B_3)$ is a $3$-premolecule, and \emph{(ii)}  $V_{t_1} \cap V_{t_2} \subset B_3$.
\end{corollary}
\begin{proof}
\cutout{  To avoid no interesting cases, we assume that the torsos $\tau_{t_1}$ is neither  a bond nor  a $3$-clique.} $\textrm{ }$ \\
\emph{(i)}. Since $\tau_{t_1}$ is $3$-connected, then by Lemma \ref{Ent:connc:lemma1}, we get $\Ent{\tau_{t_1}} \ge 3$. By Lemma \ref{torsos:minor:Coro}, the torsos $\tau_{t_1}$  is a minor of $G$, hence $\Ent{\tau_{t_1}}\le 3$.  Hence, $\Ent{\tau_{t_1}}=3$. Since $\tau_{t_1}$ is $3$-connected and and $\Ent{\tau_{t_1}}=3$, then it follows from  Lemma   \ref{3Conn:3Entang:3Mol:Lemma}, that there exists $B_3 \subset V_{\tau_{t_1}}$ such  the pair $(\tau_{t_1},B_3)$ is a $3$-premolecule.\\
\emph{(ii)}. We need the Claim:
\begin{claim}
Let $\set{v_1v_2}=V_{t_1} \cap V_{t_2}$, then  the graph $\tau_{t_1} +_{v_1v_2} C_3$ is a minor of $G$.
\end{claim}
The Claim implies that $\Ent{\tau_{t_1} +_{v_1v_2} C_3} \le 3$, then from Lemma \ref{2SumOn3Connect:Lemma2} we get that $v_1,v_2$ should belong to a minimal  edge cover of $\tau_{t_1}$. Let $B'_3$ be such a minimal  edge cover, then it follows that the pair $(\tau_{t1},B'_3)$ is again  a $3$-premolecule. 
\end{proof}

\subsection{On the diameter of the tree decomposition.}

Now we shall give conditions on the diameter of the tree decomposition if the graph has entanglement $3$.  To  give   an upper  bound for the diameter of the tree,  we have noticed that a  domino of a prescribed length is  a typical excluded minor.
The \emph{domino} $D_n$ is the graph  of vertices $V_{D_{n}}=\set{v_i,w_i\;|\; i=0,\dots,n}$  and edges $E_{D_{n}}=\set{v_iv_{i+1},w_iw_{i+1}\;|\;i=0,\dots,n-1}\cup\set{v_iw_i\;|\;i=0,\dots,n}$. 

 In the following,  if $n$ is even  then we prefer that the set of  vertices of the $n$-domino   would be of the form   
$$
\set{v_i,w_i\;|\; i=-\frac{n}{2},\dots,0,\dots,\frac{n}{2}}
$$
The $14$-domino is depicted in Figure \ref{domino:fig}.

\begin{figure}[h]
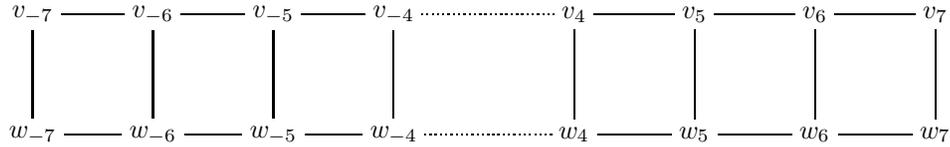

 $$
  \xygraph{ %
   !{<0cm,0cm>;<0.8cm,0cm>:<0cm,0.8cm>::}
    []*+{v_{-7}}="v0"
    [rr]*+{v_{-6}}="v1"
    [rr]*+{v_{-5}}="v2"
[rr]*+{v_{-4}}="v3"    
[rrr]*+{v_{4}}="v5"
[rr]*+{v_{5}}="v6"
[rr]*+{v_{6}}="v7"
[rr]*+{v_{7}}="v8",
[d]*+{w_{-7}}="w0"
[rr]*+{w_{-6}}="w1"
[rr]*+{w_{-5}}="w2"
[rr]*+{w_{-4}}="w3"    
[rrr]*+{w_{4}}="w5"
[rr]*+{w_{5}}="w6"
[rr]*+{w_{6}}="w7"
[rr]*+{w_{7}}="w8"
"v0"-"v1"
"v1"-"v2"
"v2"-"v3"
"v3"-@{..}"v5"
"v5"-"v6"
"v6"-"v7"
"v7"-"v8"
"w0"-"w1"
"w1"-"w2"
"w2"-"w3"
"w3"-@{..}"w5"
"w5"-"w6"
"w6"-"w7"
"w7"-"w8"
"v0"-"w0"
"v1"-"w1"
"v2"-"w2"
"v3"-"w3"
"v5"-"w5"
"v6"-"w6"
"v7"-"w7"
"v8"-"w8"
 } 
 $$
  \label{domino:fig}
  \caption{The $14$-Domino $D_{14}$}
\end{figure}

\begin{lemma}\label{Domino:Entang}
 The entanglement of the domino $D_{14}$ is at least $4$. 
\end{lemma}
\begin{proof}
We shall describe a winning strategy for Thief in $\Ent{D_{14},3}$.
First, let us fix some terminology and notations. Let $V=\set{v_i \;|\; i=-7,\dots,7}$ and $W=\set{w_i\;|\; i=-7,\dots,7}$.  Let $\LD$, \emph{the left  $4$-domino}, be the subgraph of $D_{14}$ induced by the vertices $\set{v_i,w_i\;|\; i=-7,\dots,-3}$. Similarly, we let $\RD$, \emph{the right $4$-domino}, be the subgraph of $D_{14}$ induced by the vertices $\set{v_i,w_i\;|\; i=3,\dots,7}$. We let also $\CD$, \emph{the central $2$-domino}, be the subgraph induced by the vertices $\set{v_i,w_i\;|\; i=-1,0,1}$.  If $S$ is a sub-domino of $D_{14}$ then the $S\cap V$-path (resp. $S\cap W$-path)  is  the path induced by vertices $V_{S}\cap V$ (resp. $V_S \cap W$).\\
Before giving Thief's winning strategy, we describe it informally:
\begin{enumerate}[{Step } 1.]
\item Thief plays on $\CD$ forcing Cops to place $3$ cops on it. According to the last position of Thief in $\CD$, Thief chooses  the left or the right  $4$-domino, and moreover chooses an extremal  vertex of it among $\set{v_i,w_i\;|\; i=-3,3}$. Let us assume that Thief chooses  $\LD$ and the vertex $v_{-3}$, 
\item  from the choice  $(v_{-3},\LD)$, Thief iterates moves on  the  $4$-path $\LD \cap V$  starting from $v_{-3}$ until $2$ cops are placed on this path, and at this moment 
\begin{enumerate}[{2.}1.] 
\item if there is no cop on $\CD$, then Thief goes to $\CD$, and then iterates the strategy from  Step 1, 
\item if there is a cop on the $\CD$, then there is no cop on $\RD$, and hence Thief chooses a path  to $\RD$ and  an  extremal vertex of it among $\set{v_3,w_3}$, and then  iterates the strategy, up to symmetry,  from Step 2.
\end{enumerate}
\end{enumerate}
The formal proof is split into two parts.  The first part, called the \emph{intra-steps}, consists in proving that the strategy for Thief described beside each step is   realizable. The second part, called the \emph{inter-steps}, consists in proving that the passage between steps is possible.  Technically speaking, the inter-steps proofs, are devoted to prove that the path - that  leads to the desired  sub-domino- is  free of Cops.    

Let us begin by Step 1. To argue  that Thief can play on $\CD$ in such a way he forces $3$ cops to be placed on it, it is sufficient to  observe  that $\Ent{\CD}>2$ because $\CD$ does not belong to the class $\zeta_2$ of the graphs of entanglement $2$, see \cite{BelkSanto0a7}.  The following Claim ensures that the passage from Step 1 to Step 2 is possible. 
\begin{claim}
Thief has a strategy to play in $\CD$ in such a way, once $3$ cops are placed on $\CD$, he can go in a  horizontal way   either to $\LD$ or $\RD$.
\end{claim}
\begin{proof}
Assume that Thief is trapped in  $\CD$ without being   able to find a path neither to  $\LD$ nor to $\RD$. This implies that Thief is on $v_0$ or on $w_0$ -- otherwise he is able  go to $\LD$ or to $\RD$--  and moreover he is surrounded by $3$ Cops in such a way he can not move down, left and right.  That gives rise to   the final position  $(v_{0}, \set{v_{-1},w_{0},v_{1}},Thief)$ \footnote{Since we reason up to symmetry, we put the second position $(w_0,\set{w_{-1},v_0,w_{1}},Thief)$ in the back ground.}. Coming back two moves before, we get a  position  of the form $(x,\set{v_{-1},w_0,v_1},Thief)$ where $x \in \set{v_{-1},w_0,v_1}$. From the latter position Thief is clearly able to go to either $\LD$ or $\RD$.  
This ends the proof of the Claim.
\end{proof}
Assume now that $3$ cops are placed on $\CD$ and Thief is on vertex $v_{-1}$. The other cases are handled by symmetry.  It follows from the Claim that Thief is able to choose   the pair $(v_{-3},\LD)$ by going through the path $v_{-1}v_{-2}v_{-3}$.

Let us prove  that the strategy described in  Step 2 is possible.  Once the pair $(v_{-1},\LD)$ is chosen and Thief is on $v_{-1}$, then Thief restricts his moves on the  $\LD \cap V$-path,  the latter is    of length $4$ and therefore it has entanglement $2$. Hence Thief has a strategy to force $2$ cops to be placed on this path. At this moment either there is a cop on $\CD$ or not.
\begin{itemize}
\item If there is a cop on $\CD$, and since $2$ cops are placed on the $\LD \cap V$ path, then  there is no cop on $\RD$, and moreover  there is a free path leading to one of its left extremal points  from the current vertex i.e. 
\begin{enumerate}[{(a)}]
\item if there is a cop on $\CD \cap W$, then Thief goes from $v_p\in \LD \cap V$ to $v_3$ through the free path $v_pw_pw_{p+1}\dots w_{-2}v_{-2}v_{-1}\dots v_3$.  Therefore iterates the strategy from Step 2 out of the pair $(v_3,\RD)$,  
\item if there is a cop on $\CD \cap V$, then Thief goes from $v_p \in \LD \cap V$ to $w_3$ through the free path $v_pw_pw_{p+1}\dots w_{3}$.  Therefore, iterates the strategy from Step 2 out of the pair $(w_3,\RD)$   
\end{enumerate}
\begin{remark}\label{come-back-free}
We emphasize that in case \emph{(a)}, the path chosen by Thief does not pass through $(\RD \cap W) \cup \set{w_2}$, the latter vertices are indeed free of cops and hence they might be used later by Thief. Also, in the case \emph{(b)}, the path chosen by Thief does no pass through $(\RD \cap V) \cup \set{v_2}$, the latter vertices are free of cops, and they would be used latter by Thief.   
\end{remark} 
\item if there is no cop on $\CD$ then, from the current vertex $v_p \in \LD \cap V$   Thief  goes to $\CD$ through the  path  $v_pw_pw_{p+1}\dots w_{-1}$.  The freennes of this path is ensured by   Remark \ref{come-back-free}. From the latter position, iterates the strategy from Step 1
\end{itemize} 
The strategy for Thief in $\Ent{D_{14},3}$ described so far can be iterated infinitely often, that is, it is a winning strategy for Thief. Therefore $\Ent{D_{14}}\ge 4$. 
\end{proof}
Given a graph $G$ and its Tutte's decomposition tree $T$,  we call the \emph{spread} of a vertex $v \in V_G$, denoted by $\delta_v$, the number of hinges which $v$ belongs to; i.e.
\begin{align*}
\delta_v = \card{ \set{  \set{v,x} \; \tst \; \set{v,x} \textrm{ is a hinge }  }}.
\end{align*} 
The spread of  the $2$-connected graph $G$, denoted by $\delta_G$, is defined by 
\begin{align*}
\delta_G= max \set{ \delta_v\; \tst \; v\in V_G}. 
\end{align*}
\begin{definition}\label{parallel:hinges:def}
A sequence of $n$  hinges $\set{x_1,y_1},\dots,\set{x_n,y_n}$  is  \emph{parallel} if \emph{(i)} each two hinges are  disjoint, and \emph{(ii)} there is a path $t_1\dots t_m$ in $T$ such that for each $i\in\set{1,\dots,n}$ there exists $j\in \set{1,\dots,m}$ where $\set{x_i,y_i} \in V_{t_j}$.
\end{definition}
For instance, an $n$-domino where $n\ge 2$ contains $n-1$ parallel
hinges.
\\
The following Lemma emphasizes an important aspect of hinges, intuitively it states that  hinges \emph{do not cross each other}.

\begin{lemma}\label{Hinge:Not:Cross}[Theorem IV.22 of \cite{TutteGraphTh} or Lemma 8 of \cite{Richter04}]\\
Let $G$ be a $2$-connected graph and let $\set{x_1,y_1}$ be a  hinge of $G$ and $(U,W)$ be the  $2$-separation of $G$ such that $U \cap W = \set{x_1,y_1}$. If $\set{x_2,y_2}$ is  an other hinge of $G$ such that $\set{x_1,y_1}\cap \set{x_2,y_2}=\emptyset$, then either $x_2,y_2 \in U$ or $x_2,y_2 \in W$.
\end{lemma}
\begin{lemma}\label{nhinge:ndomino}
Let $G$ be a $2$-connected graph and  $T$ its Tutte's tree. For all $n\ge 2$, if $G$ contains $4n$ parallel hinges  then $G$  has the  $(n-1)$-domino as a minor.
\end{lemma}
\begin{proof}
Let $H$  be a   set of  $\z$ hinges of $G$.
First, we  construct  a graph  $\on{G}$ out of $G$ such that $\on{G}$ is  a minor of $G$ and second,  we prove that  $\on{G}$ admits the $(n-1)$-domino as a minor. \\
Lemma \ref{howmanybridges} states that a $2$-separator $\set{x,y}$ is a hinge if and only if \emph{(i)} either there are at least  three  $[x,y]$-bridges, or \emph{(ii)} there are two  $[x,y]$-bridges at least one of them is $2$-connected.  The construction of $\on{G}$ follows. For each hinge $\set{x,y} \in H_n$, where there are at least three $[x,y]$-bridges, do the following operations:  (i) delete from $G$ all the $[x,y]$-bridges which do not   contain a hinge in $H$ other than the edge $xy$ \footnote{Observe that the hinges of  $H$ belong to at most two  $[x,y]$-bridges of $G$.}, and (ii) add the edge $xy$ to $G$.
\begin{fact} \label{fact:on:G}
The graph $\on{G}$ constructed above has the following properties:
\begin{enumerate}[(1)]
\item $\overline{G}$ is $2$-connected,
\item $\overline{G}$ is a minor of $G$, 
\item the set $H$  is again a set of hinges of $\overline{G}$, and
\item  for every $\set{x,y} \in H$, there are exactly two $[x,y]$-bridges of $\on{G}$. 
\end{enumerate}
  \end{fact} 
\begin{proof}
We argue next that the operations (i) and (ii) cited above preserve the properties stated in statements (1),(2),(3), and (4).   The statement (1) follows from  Lemma \ref{sum:lemma:bicconcted} which states that  if $\set{x,y}$ is a $2$-separator of $G$ into $(U_1,U_2)$ then $G[U_i]+ xy$, for $i=1,2$, is also $2$-connected. Statement (2) follows from Lemma \ref{Minor:2Sum}  which states that $G[U_i]+xy$ is  a minor of $G$, for $i=1,2$. Let us show statement (3). Let $\set{x,y}$ be a hinge in $H$ where $xy\in E_{\on{G}}\setminus E_G$. On the one hand, $\set{x,y}$ is again a $2$-separator in $\on{G}$.  On the other hand, one of the $[x,y]$-bridges is $2$-connected, Lemma \ref{sum:lemma:bicconcted},  because it contains the edge $xy$. Finally, statement (4) holds obviously by construction.     
\cutout{Statement (1) follows from a  straightforward induction on $\on{H}$ where $\on{H}=\set{\set{x,y} \in H \tst xy \in E_{\on{G}}\setminus E_{G}}$.   The induction  uses Lemma \ref{sum:lemma:bicconcted} which states that  if $\set{x,y}$ is a $2$-separator of $G$ into $(U_1,U_2)$ then $G[U_i]+ xy$, for $i=1,2$, is also $2$-connected. Lemma \ref{Minor:2Sum}  states that $G[U_i]+xy$ is  a minor of $G$, this allows  to argue by induction on $\on{H}$  that $\on{G}$ is a minor of $G$, i.e. statement (2). Let us  show statement  (3). Obviously, a hinge of $G$ which belongs to  $H \setminus \on H$ is  still a hinge of $\on{G}$. Let us argue about the hinges in $\on{H}$.   On the one hand each hinge $\set{x,y}$  in  $\on{H}$ is still a $2$-separator in $\on{G}$. On the other hand, for all $\set{x,y}\in \on{H}$ we have  $xy \in E_{\on{G}}$, this implies that one of the $[x,y]$-bridges in $\on{G}$ is  $2$-connected,  Lemma \ref{sum:lemma:bicconcted}. Hence,  from Lemma \ref{howmanybridges}  it follows that $\set{x,y}\in \on{G}$ is a hinge of  $\on{G}$.  Finally, statement (4) holds by construction. } 
This ends the proof of the Fact. 
\end{proof}
Second, let us prove   that $\on{G}$ admits the  $(n-1)$-domino as a minor. Let $H=\set{\set{a_1,b_1},\dots,\set{a_{\z},b_\z}}$,  and $v$ ( resp. $w$) be a vertex in the $[a_1,b_1]$-bridge (resp. in the $[a_{\z},b_{\z}]$- bridge) which  does not contain the remaining hinges. Let $\pi_x, \pi_y$ be two (simple) disjoint $v$-$w$ paths in $\on{G}$.  Such paths do exist because $\on{G}$ is $2$-connected, statement (1) of the  Fact \ref{fact:on:G}. We claim that, for each hinge $\set{a_i,b_i}\in H$, either $a_i \in \pi_x$ and hence $b_i\in \pi_y$ or $a_i\in \pi_y$ and hence $b_i \in \pi_x$. The argument is that each hinge in $H$ is a $2$-separator of $v$ from $w$ in $\on{G}$.   Therefore,  from now we  assume that the set $H$ of hinges is   of form $H=\set{\set{x_i,y_i},i=1,\dots,\z}$ such that $x_i\in \pi_x$ and $y_i \in \pi_y$, for $i=1,\dots,\z$.   Intuitively, the two paths  $\pi_x$ and $\pi_y$ would play  the role  the horizontal  lines of the  domino \footnote{if the domino is viewed in an horizontal way  as depicted  in Figure \ref{domino:fig}.} in question.  The main remaining technical part is to construct the vertical lines.        
\begin{definition}
Let $i\in\set{1,\dots,\z-1}$, and let $G_i\subset \on{G}$ be the $[x_{i+1},y_{i+1}]$-bridge   which contains the hinge $\set{x_i,y_i}$ and let $G_{i+1}\subset \on{G}$ be the $[x_{i},y_{i}]$-bridge which contains the hinge $\set{x_{i+1},y_{i+1}}$. We call an $(i,i+1)$-\emph{block} the graph $G_i\cap G_{i+1}$. If no  confusion will arise, we will call an  $(i,i+1)$-block  with simply   a block. 
\end{definition}
\begin{fact}
If one of $G_i$ and $G_{i+1}$ (given in the previous definition) is $2$-connected then the $(i,i+1)$-block is connected. Moreover, $\on{G}$ contains at least $2n$ connected blocks.
\end{fact}
\begin{proof}
We prove  the first statement of the Fact by assuming  that $G_i$ is $2$-connected. If   $G_{i+1}$ is $2$-connected then the proof is similar. Let us denote by $\mathcal{B}$ the  $(i,i+1)$-block, and  let $v_1,v_2\in V_{\mathcal{B}}$.  We shall show the existence of at least one   $v_1$-$v_2$ path in $\mathcal{B}$. Since $G_i$ is $2$-connected, then there are at least two  disjoint (simple)  $v_1$-$v_2$ paths in $G_i$. Let $\pi_1$ and $\pi_2$ be such paths. If  both $\pi_1$ and $\pi_2$  are in $\mathcal{B}$ then we have done. Otherwise, assume that $\pi_1$ does no belong to $\mathcal{B}$, i.e. $\pi_1$ contains a proper subpath which belongs to $G_i\setminus \mathcal{B}$. We claim that $\pi_1$ contains both $x_i$ and $y_{i}$ because $\set{x_i,x_{i+1}}$ is a $2$-separation in $G_i$ and the path $\pi_1$ is supposed to be simple. We mean that if $\pi_1$ visits $x_{i}$, and visits some vertices in  $G_i \setminus \mathcal{B}$, then it must visit $y_i$.    Therefore, $\pi_1$ may be of the form: $$\pi_1=v_1\dots v_px_{i}w_1 \dots w_l y_{i}v_{p+1}\dots v_{q}v_2,$$ 
or of the form:
 $$\pi_1=v_1\dots v_py_{i}w_1 \dots w_l x_{i}v_{p+1}\dots v_{q}v_2,$$ 
where $v_j \in V_{\mathcal{B}}$ for $ j=1,\dots q$ and $w_j \in V_{G_i}\setminus V_{\mathcal{B}}$ for $j=1,\dots,l$.
Since $\pi_1$ contains both $x_i$ and $y_i$, then $\pi_2$ contains neither $x_i$ nor $y_i$, because $\pi_1$ and $\pi_2$ are supposed to be disjoint. Therefore $\pi_2$ belongs to $\mathcal{B}$, i.e. the block $\mathcal{B}$ is connected.\\
To prove the second statement of the Fact, recall that from the statement (4) of the Fact \ref{fact:on:G} we have the number of $[x_j,y_j]$-bridges equals $2$, for each hinge $\set{x_j,y_j}$ in $H$. From the statement (ii) of Lemma \ref{howmanybridges} it follows that one the two $[x_j,y_j]$-bridges is $2$-connected.  This implies that one of the $(j-1,j)$-block and $(j,j+1)$-block is connected. Hence, $\on{G}$ does not contain two contiguous\footnote{Two blocks $(i,i+1)$ and $(j,j+1)$  are contiguous if $i=j+1$ or $j=i+1$.} blocks which are both not  connected. Therefore, $\on{G}$ contains at least $\frac{4n}{2}$ connected blocks.  
This ends the proof of the Fact.
\end{proof}
According to the previous Fact, $\on{G}$ contains at least $2n$ connected  blocks. Therefore, $\on{G}$ contains at least $n$ non contiguous and connected blocks $\mathcal{B}_1,\dots,\mathcal{B}_n$, the latter are  pairwise disjoint. Each block $\mathcal{B}_i,\; i=1,\dots,n$, contains a hinge $\set{a_i,b_i} \in H$, and moreover it contains a $a_i$-$b_i$ path $\pi_{i}$ because   it is connected. On the  one hand, the  graph consisting of the paths $\pi_i,i=1,\dots,n$, $\pi_x$,  and $\pi_y$ is a subgraph of $\on{G}$. On the other hand,  the $\pi_i$ paths, for $i=1,\dots,n$, are pairwise disjoint. By contracting each path  among $\pi_i$, $\pi_x$, and $\pi_y$ in the desired way we get an $(n-1)$-domino, see Figure \ref{Ndomino:Minor:Fig}. 

\begin{figure}[h]
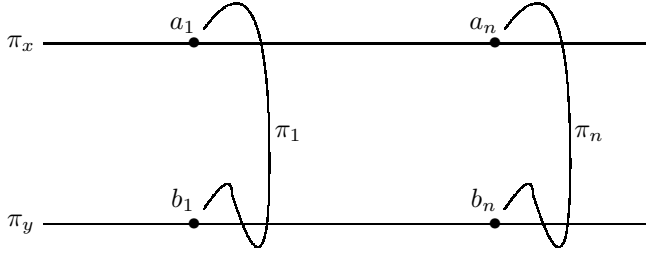

$$
  \xygraph{ %
  !{<0cm,0cm>;<1cm,0cm>:<0cm,0.8cm>::}
    []*+{}="a0"
  [rr]*+{\bullet}="a1"
    [rrrr]*+{\bullet}="an"
    [rr]*+{}="anplus1"
    [ddd]*+{}="bnplus1"
    [ll]*+{\bullet}="bn"
    [llll]*+{\bullet}="b1"
    [ll]*+{}="b0"
    "a0"[l(0.2)]="pix"
!{(-0.3 ,1)}*+{\pi_x}
!{(-0.3 ,-2)}*+{\pi_y}
    "a1"[r(1)u(2)]="a2",
      "a1"[r(1) d(2)]="a3"
      "b1"[r(1)d(1.5)]="a4"
      "b1"[r(0.5)u(0.5)]="a5"
      "b1"[r(0.5)u(1)]="a6"
"an"[r(1)u(2)]="b2",
      "an"[r(1) d(2)]="b3"
      "bn"[r(1)d(1.5)]="b4"
      "bn"[r(0.5)u(0.5)]="b5"
      "bn"[r(0.5)u(1)]="b6"
!{(1.85 ,1.3)}*+{\small{a_1}}
!{(1.85 ,-1.6)}*+{\small{b_1}}
!{(5.85 ,1.3)}*+{\small{a_n}}
!{(5.85 ,-1.6)}*+{\small{b_n}}
!{(3.25,-0.5)}*+{\small{\pi_1}}
!{(7.25,-0.5)}*+{\small{\pi_n}}
    "a0"-"anplus1"  "b0"-"bnplus1" 
"a1"-@`{..}@`{"a2"}"a3"
"a3"-@`{..}@`{"a4"}"a5"
"a5"-@`{..}@`{"a6"}"b1"
"an"-@`{..}@`{"b2"}"b3"
"b3"-@`{..}@`{"b4"}"b5"
"b5"-@`{..}@`{"b6"}"bn"
}
$$
\caption{The construction of the $(n-1)$ domino as a minor }
\label{Ndomino:Minor:Fig}
\end{figure}
 The latter has a vertical edge $a_ib_i$ for each hinge $\set{a_i,b_i}$. This ends the proof of Lemma \ref{nhinge:ndomino}.
\end{proof}

The following Lemma establishes the relation between the spread, the diameter,  and the number of parallel hinges.
\begin{lemma}\label{hugeilplieshuge:Lemma} 
Let $G$ be a $2$-connected graph and $T$ be  its Tutte's tree. Let $\delta_G$ be the spread of $G$. For all $n\ge 2$, if  the diameter of $T$ is greater than  
$4n.\delta_G$ then $G$ contains at least $n$ parallel hinges. 
\end{lemma}
\begin{proof}
 Consider  a path in $T$ of length  $\z.\delta_{G}$ of  the form $$t_1t'_1\dots t_{2n.\delta_G}t_{2n.\delta_G}'t_{2n.\delta_G +1}$$  where $t_1$ is a leaf.
One of Tutte's tree properties, see Lemma \ref{Prop:Tree}, ensures that $T$ enjoys the following properties:
\begin{enumerate}[(a)]
\item  $T$ is bipartite    and moreover every node  $t'_i, i=1,\dots,2n.\delta_G$ is such $V_{{t'_i}}$ is a hinge.  \\
\item  if $h_i$ denotes the hinge $V_{t'_i}$ then
\begin{enumerate}[({b.}1)]
\item for all $i=1,\dots,2n.\delta_{G} -1$, we have either $h_i \cap h_{i+1}=\emptyset$ or $|h_i \cap h_{i+1}|=1$, 
\item    if $h_p \cap h_q =\emptyset$ where $1\le p <q\le 2n.\delta_G$,  then for all $i \le p$ we have that $h_i \cap h_q=\emptyset$, and 
\item if $h_p \cap h_q =\emptyset$ where $1\le p <q\le 2n.\delta_G$,  then for all $i \ge q$ we have that $h_p \cap h_i=\emptyset$. 
\end{enumerate}
\end{enumerate}
Let $H=\set{h_1,\dots,2n.\delta_G}$ be  a sequence of hinges where $h_i=V_{t'_i}$. Out of $H$ we shall define an  undirected graph $\spec{H}$ in such a way the properties of hinges in $H$ transfer into the properties of $\spec{H}$. Define $\spec{H}$ as follows $V_{\spec{H}}=\set{h^{v}\;|\; h\in H}$, in other words the vertex $h^{v}$ is just the hinge $h$ viewed as a vertex;  and  $h^{v}_1h^{v}_2 \in E_{\spec{H}}$ if and only if $h_1\cap h_1 \neq \emptyset$. We state the main properties of the graph $\spec{H}$.
\begin{fact}\label{This:fact:1}
Let $\mathcal{B}$ be a $2$-connected component of $\spec{H}$, and let $m=|V_{\mathcal{B}}|$, then $\mathcal{B}$ is an $m$-clique. 
\end{fact}
\begin{proof}
Let $\ver{h}_1,\ver{h}_2\in V_{\mathcal{B}}$.  Since $\mathcal{B}$ is $2$-connected then  it follows by Menger Theorem \ref{Menger:Theorem} that there exist two disjoint $\ver{h}_1$-$\ver{h}_2$ paths  in $\mathcal{B}$ i.e. there is a cycle in $\mathcal{B}$ containing both $\ver{h}_1$ and $\ver{h}_m$. Let $\ver{h}_1,\ver{h}_2\dots \ver{h}_m\ver{h}_1$ be such a cycle. Assume that  $\ver{h}_1\ver{h}_3 \notin E_{\mathcal{B}}$, hence $h_1\cap h_3 =\emptyset$. It follows  from property (b.3) above  that for all $i\ge 3$ we have $h_1 \cap h_i =\emptyset$. This is a contradiction since already $h_1 \cap h_n \neq \emptyset$ because $\ver{h}_1\ver{h}_m \in E_{\mathcal{B}}$. We deduce that $\ver{h}_1\ver{h}_3 \in E_{\mathcal{B}}$. Using similar arguments, we deduce that for all $i= 3,\dots,m-1$ we have $\ver{h}_1\ver{h}_{m-1} \in E_{\mathcal{B}}$. We conclude that for every pair $\ver{h}_1,\ver{h}_2$ of vertices in the component $\mathcal{B}$ there is an edge $\ver{h}_1 \ver{h}_1$ in $\mathcal{B}$,  implying that $\mathcal{B}$ is an $m$-clique.                 
This ends the proof of the Fact.
\end{proof}

\begin{fact}\label{This:fact:2}
If $\spec{H}$  contains an $m$-clique of  vertices $\ver{h}_1,\dots,\ver{h}_m$ then 
\begin{align*}
|\cap_{i=1,\dots,m} h_i|=1.
\end{align*}
\end{fact}
\begin{proof}
The proof is by induction on $m$.\\
If $m=3$,  then either  $\cap_i h_i\neq \emptyset$ and in this case $|\cap_ih_i|=1$, or $\cap_i h_i=\emptyset$ and in this case the hinges $h_1,h_2,h_3$ form a triangle in the following sense: $h_i \cap h_j =\set{v_{ij}}$ for $i,j=1,2,3$ and $i\neq j$.  In the latter case,  the hinges $h_1,h_2,h_3$ do not belong to a unique path in $T$, contradicting the hypothesis that the hinges in question are parallel, condition \emph{(ii)} of Definition \ref{parallel:hinges:def}.   \\
Induction step. Consider a $(m+1)$-clique of vertices $\ver{h}_1,\dots,\ver{h}_{m+1}$. By induction hypothesis $|\cap_{i=2,\dots,{m+1}}h_i|=1$ and $|\cap_{i=1,\dots,m}h_i|=1$. This implies that $\cap_{i=2,\dots,{m+1}}h_i = \cap_{i=1,\dots,{m}}h_i$, and hence   $|\cap_{i=1,\dots,m+1}h_i|=1$.  This ends the proof of the Fact.
\end{proof}
Summing up the Facts \ref{This:fact:1} and \ref{This:fact:2} we deduce that the size (i.e. the number of vertices)  of the $2$-connected components of $\spec{H}$ is at most $\delta_G$, because, given a $2$-connected component $\mathcal{B}$ of vertices $\ver{h}_1,\dots,\ver{h}_m$   the spread  of the vertex $\cap_{i=1,\dots,m}{h_i}$ is $m$.  The following  Claim  allows us to compute a lower bound of the number of $2$-connected components of $\spec{H}$.
\begin{claim}
Let $H$ be a graph. If the size (i.e. number of vertices) of its $2$-connected components is at most $\delta$ then $G$ has got at least $\lfloor \frac{|V_H|}{\delta}\rfloor$ $2$-connected components.
\end{claim}
\begin{proof}
Let $x$ be the number of $2$-connected components of $H$. Clearly $|V_H|\le x.\delta$, hence $\frac{|V_G|}{\delta} \le x$. Therefore $H$ contains at least $\lfloor \frac{|V_H|}{\delta}\rfloor$ $2$-connected components. This ends the proof of the Claim.
\end{proof}
Since $|V_{\spec{H}}|=2n.\delta_G$ and the size of the $2$-connected components of $\spec{H}$ is at most $\delta_G$, then, according to the Claim,    the number of the  $2$-connected components of $\spec{H}$ is at least $2n$. \\ 
On the one hand, if the  two vertices $\ver{h}_1$ and $\ver{h}_2$ belong to two disjoint $2$-connected components then the related hinges $h_1$ and $h_2$ are disjoint, because otherwise, there is an edge $\ver{h}_1\ver{h}_2$ in $\spec{H}$ implying that both $\ver{h}_1$ and $\ver{h}_2$ belong to the same $2$-connected component, which is a contradiction. One the other hand $\spec{H}$ contains at least $n$ pairwise disjoint $2$-connected components.  We conclude that $G$ contains at least $n$ disjoint hinges, the latter belong  to  the same path in $T$, and therefore they are parallel.  This ends the proof of Lemma \ref{hugeilplieshuge:Lemma}.    
\end{proof}

Now we are ready to provide an upper bound of the diameter of Tutte's tree if the related graph has entanglement $3$.

\begin{proposition}\label{Diameter:entag:3}
Let  $G$ be a $2$-connected graph,  $T$ be its Tutte's tree,  and $\delta_G$ be   the spread  of $G$.  If the   entanglement of $G$ is $3$ then the diameter of $T$ is at most $2^{7}.\delta_G$ . 
\end{proposition}
\begin{proof}
Let $G$ be as stated in the Proposition  with  $\Ent{G}=3$. If the diameter of $G$  is  $c .\delta_G $ then  it follows from the   Lemma \ref{hugeilplieshuge:Lemma} that $G$ contains at least $\lfloor c\setminus 4\rfloor$ parallel hinges.  Therefore, from   Lemma \ref{nhinge:ndomino}  we get that $G$ contains a $c'$-domino $D_{c'}$ as a  minor, where $c'=\lfloor \lfloor c \setminus  4 \rfloor\setminus 4\rfloor-1$. Hence $ \Ent{D_{c'}} \le \Ent{G}$ by Theorem \ref{minorclosureTh} which states that the class of undirected graphs of entanglement at most $k$ is minor ideal. From Lemma \ref{Domino:Entang}, a $14$-domino has entanglement strictly greater than $3$. We conclude that $c' \le 14$ and hence $c \le 2^{7}$. 
\end{proof}

Now we are ready to state the main result of this paper. To this goal, we define the \emph{interface} of a torso $\tau_t$ as $$I_t=\set{v \in V_t \;|\; \exists  tt' \in E_T \textrm{ s.t } v\in V_{t'}}. $$
 
\begin{theorem} \label{Main:Theorem}
Let $G$ be a $2$-connected graph  and let  $({T,(V_t)}_{t\in T})$ be the Tutte decomposition of $G$.  If the  entanglement of $G$ is  $3$   then  for every $3$-connected torso $\tau_{t}$ of $G$  the following hold: 
\begin{enumerate}[1.]
\item  there exists $B_3 \subseteq V_t$ where $|B_3|=3$  such that $(\tau_t,B_3)$  is an abstract $3$-molecule, 
\item  $I_t \subseteq  B_3$, and 
\cutout{\item $deg_{T}(t) \le 3$, and}
\item if  $\delta_G$ is the  spread of $G$, then  the diameter of $T$ is at most $2^{7}.\delta_G$.
\end{enumerate}
\end{theorem}
\begin{proof}
Let $G$ be as stated and $\tau_t$ a $3$-connected torso  of $G$.
\begin{enumerate}
\item On the one hand,  Lemma \ref{torsos:minor:Coro} states that  $\tau_t$ is a minor of $G$.  It follows from Theorem \ref{minorclosureTh} that  $\Ent{\tau_t} \le 3$. On the other hand, since $\tau_t$ is $3$-connected then we get from Lemma \ref{Ent:connc:lemma1} that $3 \le \Ent{\tau_t} $, therefore $\Ent{\tau_t}=3$. Hence, from Lemma \ref{3Conn:3Entang:3Mol:Lemma} we deduce that there exists $B_3 \subset V_{\tau_t}$ with $|B_3|=3$ such that  $(\tau_t,B_3)$ is an abstract  $3$-molecule. 
\item  If the $3$-molecule $\tau_t$ is not ambiguous, then the property $I_t\subseteq B_3$ follows from statement \emph{(ii)} of Corollary  \ref{Interface:Cor}. If $\tau_t$ is not ambiguous, then it suffices to generalize Lemma \ref{2SumOn3Connect:Lemma2} as follows:
\begin{lemma}
  Let $\vartheta$ be a $3$-molecule, and let each $C_3,C'_3,C_3$ and $C''_3$ be a  $3$-cycle. Let 
$$
G= ((\vartheta +_{v_1w_1} C_3)+_{v_2w_2}C'_3) +_{v_3w_3} C''_3, 
$$
where $+$ is the $2$-Sum operator and $v_iw_i \in E_{\vartheta},\; i=1,2,3$, $v_1w_1 \in E_{C_3}$, $v_2w_2 \in E_{C'_3}$, and $v_3w_3 \in E_{C''_3}$.
If the entanglement of ${G}$ is again $3$ then this implies that the vertices $v_i,w_i,i=1,2,3$ belong to a minimal edge cover set of $\vartheta$.  
\end{lemma}  
\item The condition on the diameter of $T$ follows from Proposition \ref{Diameter:entag:3}.
\cutout{\item For some $t'$ such that $tt'\in E_T$, we get from Lemma \ref{2SumOn3Connect:Lemma}.3 that $V_t \cap V_{t'} \subset \set{a,b,c}$  such that $\set{a,b,c}$ is a minimal feedback vertex set of some isomorphic representation of the $3$-molecule  $\tau_t$. Hence $I_t \subseteq \set{a,b,c}$.
\item One can reformulate item $2$ as follows: for every hinge $[x,y]$ of $G$ such that $x,y$ belong to a $3$-connected torsos $\vartheta_{a,b,c}^{A,h}$, we have $x,y \in \set{a,b,c}$. Let $\mathcal{B}_{x,y}$ be the set of all torses having $[x,y]$ as a hinge. Let ${deg}_{x,y}(t)$ be the number  of the neighbours of $t$ that admits $[x,y]$ as a hinge. Our aim is to compute 
\begin{equation}\label{som:deg:eq}
{deg}_{T}(t)={deg}_{a,b}(t) + {deg}_{b,c}(t) + {deg}_{a,c}(t)
\end{equation}
To this goal, we distinguish many cases according to the nature of the (eventual) edge $xy$ of $\tau_t=\vartheta_{a,b,c}^{A,h}$.
}

\end{enumerate}
\end{proof}

\section{Towards sufficient conditions on Tutte's tree}
 Theorem \ref{Main:Theorem} provides some necessary conditions on Tutte's tree to be a tree decomposition of a $2$-connected graph of entanglement $3$. However, these conditions are not sufficient, since the $14$-domino is a counter example. Technically speaking, the $14$-domino is not a worth counter example because the computation of the exact value of Tutte's tree would allow us to obtain sufficient conditions.  To be more precise,  the exact value of the diameter would essentially depend on, besides the spread,  the length  of each cycle  which constitutes a torso of Tutte's tree.  In the sequel  we shall tell something about the entanglement of $2$-connected graphs for which Tutte's tree satisfies conditions (1) and  (2) of Theorem \ref{Domino:Entang} and the hinges of $G$  are organized into a path-like structure in the following sense: there exists a path $t_1,\dots,t_n$  in $T$ such that for each hinge $h$ of $G$ there exists $t_i$ where $h=V_{t_i}$, $i \in \set{1,\dots,n}$. 
\begin{proposition}\label{Entag:Path}
Let $G$ be a $2$-connected graph and $T$ be its Tutte's tree.   If $T$ satisfies conditions (1) and (2) of Theorem \ref{Domino:Entang} and if the hinges of $G$ are organized into a path-like structure, then $G$ has entanglement at most $4$. 
\end{proposition}
\begin{proof}
The proof of This Proposition is essentially an adaptation of the  proof of Proposition 8.8.1 of \cite{mathese}, the latter states that  the entanglement of the $n$-domino is at most $4$.   

Let $G$ and $T$ be as stated in the Proposition. We shall describe a winning  strategy for Cops in $\Ent{G,4}$.\\
  Let $H=\set{ h_i\;|\; i=0,n}$ be the sequence  of hinges of $G$.   Cops strategy in $\Ent{G,4}$ is described  by means of the following steps: 
\begin{enumerate}[{Step } 1.]
\item By playing on the hole graph and  using  just $2$ cops, occupy  a hinge of the form   $h_j$ where $0\le j \le n$. At this moment, Thief goes either   to an $h_j$-bridge which does not contain any hinge, or he goes to an  $h_j$-bridge which contain at least one hinge. Observe that there exist at most two  $h_j$-bridges which contain hinges because the bridges of $G$ are organized into a path-like structure.   In the first case Cops use the third cop to catch Thief in the $h_j$-bridge, the latter is either a path, a cycle, or a $3$-molecule by assumptions 1 and 2 of Theorem \ref{Main:Theorem}.    In the second case,   iterate the strategy from  Step 2. 
\item   By playing on the related  $h_j$-bridge --  where $2$ cops  are placed on $h_j$ --  Cops use the other $2$ cops to occupy a hinge  of form  $h_i$.   At this moment, Thief has $3$ possibilities: \emph{(1)} he goes to one of the $h_i$-bridges which  does not contain any hinge, and therefore  he will trapped, \emph{(2)} he goes to the $h_i$-bridge which dos not contain the hinge $h_j$, and in this case iterate the strategy from Step 2 (up to symmetry) by keeping the two cops on $h_i$ and using the two cops on $h_j$, or  \emph{(3)} he goes to the $h_i$-bridge  which contains  the hinge $h_j$, and in this case iterate the strategy from Step 3.     

\item By playing on the  $h_i$-bridge $\cap$ $h_j$-bridge   which contains both $h_i$ and $h_j$  -- where the $4$ cops are placed  on $h_i\cup h_j$\footnote{We can assume w.l.g that $h_i\cap h_j=\emptyset$, otherwise Thief will be easily trapped.} and Thief is on vertex $x_i \in h_i$, (the other positions of Thief are handled by symmetry) -- Cops use the two cops on $h_k$ to occupy a hinge of the  form $h_p$ where $i<p<j$ and  $k$ equals $j$ if $j\le n-i$ and  equals $i$ otherwise. Besides, Cops do not remove the $2$ cops on $h_{\overline{k}}$, where $\overline{k}=i$ if $k=j$ and $\overline{k}=j$ if $k=i$.  At this moment, i.e. from the position  $x_p\in h_p$ , Thief has three possibilities: either \emph{(1)} he goes to an $h_p$-bridge which does not contain any hinge, and therefore he will be trapped, \emph{(2)} he goes to the $h_p$-bridge which  contains hinges but not the $h_{\overline{k}}$ one,  and in this case iterate the strategy from Step 2, or \emph{(3)} he goes to the $h_p$-bridge $\cap$ $h_{\overline{k}}$-bridge and in this case iterate the strategy from Step 3. 
\end{enumerate}
It remains to argue that (i) Cops strategy described in each step is realizable and (ii) this strategy would not be iterated infinitely often. To prove  statement (i), it is enough   to prove  the following Fact.
\begin{fact}
Let $G$ be a $2$-connected graph and $T$ be its Tutte's tree.   If $T$ satisfies conditions (1) and (2) of Theorem \ref{Main:Theorem} and if the hinges of $G$ are organized into a path-like structure, then   in the game $\Ent{G,2}$  Cops have a strategy to occupy  a hinge of  $G$ or they win. 
\end{fact}   
\begin{proof}
Let $H=\set{\set{x_i,y_i}, i=0,n}$ be the sequence of hinges of $G$.

First, if Thief restricts his moves on some torso $\tau_t, t\in V_t$, then we distinguish two cases according to the nature of $\tau_t$:
\begin{proofbycases}
 \begin{caseinproof}
  if $\tau_t$ is a cycle of the form $v_1v_2\dots v_m v_1$, and Thief is walking from $v_1$ to $v_m$, then Cops strategy consists in putting the first cop on $v_1$ and following Thief with the second cop, until a vertex of some hinge $h_i \subset \set{v_1,\dots,v_m}$  is occupied by a cop, and then Cops do not move this cop and follow Thief with the other cop until the second vertex of the hinge $h_i$ is occupied by the second cop, and we have done. 
 \end{caseinproof}
 
 \begin{caseinproof}
if $\tau_t$ is  $3$- connected then it is a $3$-molecule by condition  (1) of Theorem \ref{Main:Theorem} and the interface  $I_t $ of $\tau_t$ belongs to a minimal edges cover set of $\tau_t$ by condition (2) of the same Theorem. In this case, Cops strategy consists in skipping on the vertices $V_{\tau_t} \setminus I_{t}$ and  using the two cops  to occupy two vertices in $I_{t}$.  These two vertices obviously constitute a hinge of $G$, and we have done.     
 \end{caseinproof}
\end{proofbycases}

Second, if Thief does not restrict his moves on the same torso, i.e. he moves on a sequence of torsos $\tau_{t_1}, \tau_{t_2}, \dots$ where $t_1,t_2, \dots \in V_T$, then Cops strategy consists in placing the first cop  on the vertex of $\tau_{t_1} \cap \tau_{t_2}$ which is visited by Thief (recall that $\tau_{t_1} \cap \tau_{t_2}$ is a hinge),  placing the second cop on  the vertex -- visited by Thief -- of $\tau_{t_2} \cap \tau_{t_3}$, replacing the first cop to the vertex -- visited by Thief -- of  $\tau_{t_3} \cap \tau_{t_4}$, \dots  until some torso  $\tau_{t_i}$ is visited twice by Thief. Observe that there are two cops placed on $\tau_{t_i} \cap \tau_{t_{i+1}}$. Since $\tau_{t_i} \cap \tau_{t_{i+1}}$ is a hinge of $G$, we have done. This completes the proof of the Fact. 
\end{proof}

To  show statement (ii) above it is enough to observe that whenever a step is revisited a second time, then the   value $|i-j|$ is strictly lower than that of the first time. This ends the proof of Proposition  \ref{Entag:Path}. 
\end{proof}

As a direct consequence  of the previous Proposition, the $n$-domino has entanglement at most $4$ and therefore it follows from Lemma \ref{Domino:Entang} that the entanglement of the $n$-domino, where $n \ge 14$, is exactly $4$.   \\

\paragraph{Acknowledgement.} We acknowledge Luigi Santocanale  for helpful  discussions   on the topic.

\bibliographystyle{halpha}
\bibliography{../Mathese/biblio}

\newpage

\end{document}